\documentclass[pre,aps,showpacs,longbibliography]{revtex4-1}
\usepackage[utf8]{inputenc}
\usepackage{rotating}		
\usepackage{amsmath}	
\usepackage{amssymb}	
\usepackage{color}
\usepackage[hidelinks]{hyperref}
\hypersetup{  
colorlinks=true,
citecolor=blue} 
\usepackage{bm}
\usepackage{amsthm}
\newtheorem{theorem}{Theorem}

\newcommand{\f}[2]{
		\mathchoice%
			{\dfrac{#1}{#2}}
	    	{\dfrac{#1}{#2}}
			{\frac{#1}{#2}}
			{\frac{#1}{#2}}
}
\newcommand{\dd}{\mathrm{d}}
\newcommand{\e}[1]{_{\text{#1}}}
\newcommand{\ddf}[3][]{%
        \ifthenelse{\equal{#1}{}}{%
                \ensuremath{\f{\dd#2}{\dd#3}}%
        }{%
                \ensuremath{\f{\dd^{#1}#2}{\dd{#3}^{#1}}}%
        }%
}
\newcommand{\Dp}[3][]{
        \ifthenelse{\equal{#1}{}}{%
                \ensuremath{\f{\partial#2}{\partial#3}}%
        }{%
                \ensuremath{\f{\partial^{#1}#2}{\partial{#3}^{#1}}}%
        }%
}

\newcommand{\MF}{\textsc{mf}}
\newcommand{\A}{\mathcal{A}}
\newcommand{\Hr}{\mathcal{H}}
\newcommand{\xetoile}{x^*}
\newcommand{\Imicro}{I_{\textrm{mc}}}
\newcommand{\Icano}{I_{\textrm{c}}}
\newcommand{\IN}{{\mathrm{int}}}
\newcommand{\obsmag}{{m}}
\newcommand{\obsm}{{m}}
\newcommand{\bromwich}{{b}}
\newcommand{\pstat}{\pi}

\begin{document}

\title{Non equivalence of dynamical ensembles and emergent non ergodicity}

\author{Hadrien Vroylandt}
\affiliation{Laboratoire de Physique Th\'eorique (UMR8627), CNRS, Univ. Paris-Sud, Universit\'e Paris-Saclay, 91405 Orsay, France}
\author{Gatien Verley}
\affiliation{Laboratoire de Physique Th\'eorique (UMR8627), CNRS, Univ. Paris-Sud, Universit\'e Paris-Saclay, 91405 Orsay, France}
\date{\today}

\begin{abstract}
Dynamical ensembles have been introduced to study constrained stochastic processes. In the microcanonical ensemble, the value of a dynamical observable is constrained to a given value. In the canonical ensemble a bias is introduced in the process to move the mean value of this observable.
The equivalence between the two ensembles means that calculations in one or the other ensemble lead to the same result. 
In this paper, we study the physical conditions associated with ensemble equivalence and the consequences of non-equivalence. For continuous time Markov jump processes, we show that ergodicity guarantees ensemble equivalence. For non-ergodic systems or systems with emergent ergodicity breaking, we adapt a method developed for equilibrium ensembles to compute asymptotic probabilities while caring about the initial condition. We illustrate our results on the infinite range Ising model by characterizing the fluctuations of magnetization and activity. We discuss the emergence of non ergodicity by showing that the initial condition can only be forgotten after a time that scales exponentially with the number of spins. 
\end{abstract}

\maketitle


\section{Introduction}

Ensemble equivalence offers a convenient way of computing equilibrium potentials (entropy, free energy, grand potential, etc) by choosing the ensemble in which calculations are easier. Then, a Legendre transform provides the appropriated potential according to the environmental constraints on the system (isolated, thermostatted, chemostated, etc) \cite{Callen1985}. Ensemble equivalence holds in many cases, but it breaks down for systems with long range interactions between system constituents \cite{Barre2001,Campa2009,Bouchet2010} or in the presence of a phase transition treated at the mean field level \cite{Ellis2004,Costeniuc2005}. The equivalence of equilibrium ensemble has been studied in detail \cite{Campa2009,Barre2001,Bouchet2010,Touchette2004,Bouchet2008}, including in the framework of large deviation theory \cite{Touchette2003,Touchette2009}. 

More recently, different dynamical ensembles have been used by many authors to study stochastic processes \cite{Chetrite2015,Chetrite2013,Jack2010,Lecomte2007,Speck2016,Garrahan2016,Monthus2011,Budini2014} and the question of their equivalence has been an active field of research. Equilibrium and dynamical ensembles differ qualitatively in nature: the first one is made of states and the second of succession of states (either continuous or discrete) called trajectories. Then, the probabilities on the ensembles are also different, a state probability for the former and a probability of trajectories for the latter. Dynamical ensembles must be used when the considered variables are time-integrated observables (like work, irreversible heat exchanges or activity) because their statistics can only be computed using trajectory probabilities. 
Based on a choice of dynamical observable and of a stochastic process defining bare trajectory probabilities, one can define two dynamical ensembles: the microcanonical and the canonical ensembles. Within the first ensemble, the trajectories are filtrated on the value of the chosen dynamical observable. Within the second ensemble, the probability of trajectories are exponentially biased. This canonical ensemble is also called the \textit{s-ensemble} \cite{Garrahan2009}, the \textit{driven}, \textit{biased} or \textit{tilted} ensemble \cite{Jack2010} or \textit{Esscher transform}\cite{Feller2008}. The last ensemble has been used to study glass transition \cite{Garrahan2009,Keys2015,Biroli2013,Merolle2005,Jack2006,Garrahan2007,Hedges2009,Chandler2010,Speck2011} or as a numerical tool to evaluate rare event probabilities \cite{Nemoto2014,Ferre2018}. As indicated by their names, the way one introduces dynamical ensembles is in clear analogy with the way ensembles are defined in equilibrium statistical physics.

When two ensembles are equivalent, the mean value of the dynamical observable is the same in the two ensembles. 
Like for equilibrium ensembles, it is more convenient to make calculations (or simulations) within the canonical ensemble motivating the determination of the conditions for ensemble equivalence.
Large deviation theory allows to conclude on this point from the convexity of the large deviation function (LDF) or from the differentiability of the cumulant generating function (CGF) \cite{Chetrite2013,Chetrite2015}. 
Whenever ensembles are equivalent, LDF and CGF encode the same information and are related by Legendre transform.

In this paper, we aim in the first place at exploring the physical constraints associated with equivalence of microcanonical and canonical dynamical ensembles. In the second place, when ensemble equivalence does not hold, we provide a method to compute non-convex LDFs from the (moment) generating functions. 

We identify that ergodicity and additivity are crucial for the ensemble equivalence. Ergodicity guaranties that the system can switch from any state to another in a reasonable time allowing to concatenate pieces of trajectories. Thanks to additivity, the probability of the dynamical observable on a trajectory made of two pieces is connected to a weighted average of the pieces probabilities (based on their duration) constraining the dynamical observable probability to ensure ensemble equivalence. In section \ref{sec:PathEnsembles}, we review results on ensemble equivalence and show for Markov jump processes that ergodicity implies the equivalence of dynamical ensembles for any additive observable. The contraposition says that non-ergodicity is required (but may not be sufficient) to break ensemble equivalence, and thus to obtain non-convex LDF. In section \ref{sec:expl-comp-ldf}, we illustrate our results with a simple non-ergodic system and derive the non-convex LDF for the system activity. For this system, non-ergodicity is imposed by construction because the transition rate matrix is reducible. The non-ergodicity impacts the LDF since it may change according to the probability of the initial state. Beyond this case, we look at the magnetization and activity of an infinite range Ising model for which ergodicity is broken in the thermodynamic limit. We show how to compute the non-convex LDF from the propagator of the generating function for the chosen observable(s) by adapting a method developed by Touchette in Ref.~\cite{Touchette2010,Touchette2008} for equilibrium ensembles. Applying this to the Ising model, we notice that the non equivalence of equilibrium ensembles is tightly connected to the non equivalence of the dynamical ensembles. Since the non-ergodicity of the Ising model is an emergent feature, the order of the thermodynamic and long time limits determines which of the convex or non-convex LDF is more appropriated. In section \ref{sec:LimitOrder}, we prove that the time to relax from the initial condition grows exponentially with the number of spins. In practice, this means that for time shorter than this relaxation time the system is not ergodic and one must use the non-convex LDF. For longer times of observation, the system is ergodic and one must use the convex LDF. Such arguments are conventional in the framework of equilibrium ensembles: we adapt them for dynamical ensembles and illustrate them precisely using the Ising model.

\section{Microcanonical and Canonical Dynamical Ensembles}
\label{sec:PathEnsembles}

In this section, we review for time-homogeneous jump processes some results on dynamical ensembles equivalence, see Refs.~\cite{Chetrite2013,Chetrite2015} and references therein for more general Markov processes.

\subsection{Definitions: stochastic process and dynamical ensembles}
\label{sec:general-theory}

The Markov jump process $X_t$ takes values in the finite state space $\Omega $ at all time $t \in [0,T]$. The probability that $X_t=x \in \Omega$ at time $t$, denoted $p(x,t)$, is the solution of the time-homogeneous master equation
\begin{equation}
  \label{eq:EqMaitress}
\Dp{}{t} p(x,t) = \sum_{y}  K(x,y) p(y,t)
\end{equation}
where $K(x,y)$ are the components of the Markov operator $\bm{K}$, i.e. the probability per unit time of the transition from $y$ to $x$. The Markov jump process is said to be irreducible if for any pair of states $x \neq y$ it exists a path that connects $x$ to $y$ (and $y$ to $x$), i.e. every transitions along the path have non-vanishing transition rates. On the opposite, the process is reducible when at least two subsets of states in $\Omega$ cannot be connected. Then, the process is non-ergodic. As a consequence, the ensemble average does not produce the same result as the time average since the process cannot explore the whole state space.

Alternatively, one may use the path-integral formalism to characterize the process $X_t$. In this case, we denote by $P[x]$ the path probability, where $[x]$ is the short notation for a realization of the process $X_{t}$ over the time interval $[0,T]$.

We aim at computing the statistics of a dynamical observable $\mathcal{O}$. The mean value of this observable $\mathcal{O}$ follows from
\begin{equation}
  \label{eq:meanAT}
  \langle \mathcal{O} \rangle = \int \mathcal{D} [x] P[x] \mathcal{O}[x],
\end{equation}
where $\int \mathcal{D}[x]$ denotes the sum over all possible paths. For simplicity, we consider an observable of the form
\begin{equation}
  \label{eq:AT}
\mathcal{O} \equiv \f{1}{T} \int_0^T f(X_t) \dd t + \f{1}{T} \sum_{0\leq t\leq T:\Delta X_t\neq 0}g(X_{t^+},X_{t^-})
\end{equation}
where $f(x)$ and $g(x,y)$ are arbitrary functions and the discrete sum is on the time $t$ at which the system changes of state, i.e. $\Delta X_{t} \equiv X_{t^{+}}-X_{t^{-}}\neq 0$. The $f$ dependent part enables to study the time average of a state function, and the $g$ dependent part the time average of jump quantities, like currents.  For example, choosing $f(x)=0$ and $g(x,y)=1$ will lead to the global activity of the system, i.e. the number of jumps per unit time in the trajectory. 

The \emph{microcanonical ensemble} is defined by filtering the trajectory ensemble with a condition on an observable, e.g. $\mathcal{O} = \mathfrak{o}$ where $ \mathfrak{o}$ is a possible value of the stochastic variable $\mathcal{O}$ such that its probability verifies $P_{T}(\mathcal{O}= \mathfrak{o}) > 0$. We write the path probability of these trajectories 
\begin{equation}
  \label{eq:proba-micro-path}
  P^{\mathfrak{o}}[x] = \f{P[x,\mathcal{O}= \mathfrak{o}]}{P_{T}(\mathcal{O_{T}}= \mathfrak{o})},
\end{equation}
where $P^{\mathfrak{o}}[x]$ is the conditional probability of the path $[x]$ given that for this path $\mathcal{O}= \mathfrak{o}$. The corresponding joint probability is denoted $P[x,\mathcal{O}= \mathfrak{o}]$.
In the microcanonical path ensemble, the observable $\mathcal{O}$ does not fluctuate and always achieves the same value for all trajectories in the ensemble.

The \emph{canonical ensemble} is defined by fixing the mean value of the observable $\mathcal{O}$ only. The path probability for this ensemble can be computed via tilting (or biasing) the process as follows:
\begin{equation}
  \label{eq:proba-cano-path}
  P_\gamma[x]=\f{e^{T\gamma \mathcal{O}}P[x]}{\langle e^{T\gamma \mathcal{O}}\rangle}.
\end{equation}
This path probability is normalized by construction.

\subsection{Path ensemble equivalence}

To discuss the dynamical ensemble equivalence, one needs to define when two path probabilities $P$ and $Q$ are equivalent. They do when
\begin{equation}
  \label{eq:equivalence-def}
  \lim\limits_{T \to + \infty}\f{1}{T}\ln \f{P[x]}{Q[x]} = 0
\end{equation}
almost everywhere with respect to $P$. This means that $P$ and $Q$ are equal up to subexponential terms in $T$ for almost all paths. As a consequence, the mean value at large time of any observable will be the same if computed with any one of the two equivalent path ensembles.

Assuming a large deviation principle for $\mathcal{O}$~\cite{Touchette2009}, the LDF function $I(\mathfrak{o})$ provides the rate of decay of the probability density function $P_{T}(\mathcal{O}= \mathfrak{o})$ in the limit $T\to +\infty$ such that
\begin{equation}
  \label{eq:LDF-observable}
  P_{T}(\mathcal{O}= \mathfrak{o}) \underset{T \to + \infty}{\simeq} e^{-T I(\mathfrak{o})}
\end{equation}
A LDF indicates the rate at which the probability density function of a time extensive observable concentrates in its most likely value when time goes by.

Assuming a large deviation principle, Chetrite and Touchette have shown that the convexity of $I(\mathfrak{o})$ determines whether the microcanonical and canonical path ensembles are equivalent~\cite{Chetrite2015}, in such a case $P^{\mathfrak{o}}$ and $P_\gamma$ satisfy Eq.~(\ref{eq:equivalence-def}). They distinguished three cases:
\begin{itemize}
\item[] \textbf{(Equivalence)} If $I(\mathfrak{o})$ is a stricly convex function at $\mathfrak{o}$, then there exists a unique $\gamma \in \mathbb{R}$  such that  $P^{\mathfrak{o}}$ and $P_\gamma$ are equivalent.
\item[] \textbf{(Non equivalence)} If $I(\mathfrak{o})$ is a non-convex function at $\mathfrak{o}$, then there are no $\gamma$ such that  $P^{\mathfrak{o}}$ and $P_\gamma$ are equivalent.
\item[] \textbf{(Partial equivalence)} If $I(\mathfrak{o})$ is a convex function but not strictly convex function at $\mathfrak{o}$, then numerous values of $\mathfrak{o}$ correspond to the same $\gamma$: It may correspond to linear parts in a convex function or to a point of a non-convex function at which the slope is not unique.
\end{itemize}

The convexity of the LDF $I(\mathfrak{o})$ is connected to the differentiability of the CGF. The CGF $\phi(\gamma)$ is defined as
\begin{equation}
  \label{eq:def-CGF}
  \phi(\gamma) \equiv \lim\limits_{T\to +\infty} \f{1}{T} \ln \langle e^{T\gamma \mathcal{O}}\rangle.
\end{equation}
The CGF is the highest eigenvalue of the titled matrix $\bm{K}_{\gamma}$ of elements~\cite{Lebowitz1999}
\begin{equation}
  \label{eq:tiltedmatrix}
  K_\gamma(x,y) \equiv \left\{\begin{array}{ll} K(x,y)e^{\gamma g(x,y)}  & \text{ if } x \neq y \\K(x,x)+\gamma f(x) & \text{ if } x=y \end{array} \right. .
\end{equation}
The CGF is also the Legendre-Fenchel transform of the LDF
\begin{equation}
  \label{eq:legendre-cgf}
   \phi(\gamma)= \max\limits_{\mathfrak{o}} (\gamma \mathfrak{o}-I(\mathfrak{o})).
 \end{equation}
Conversely, we define $I^{**}(\mathfrak{o})$ as the Legendre-Fenchel transform of the CGF
\begin{equation}
  \label{eq:gartner-ellis}
  I^{**}(\mathfrak{o}) \equiv \max\limits_{\gamma}( \gamma \mathfrak{o}-\phi(\gamma))= \gamma(\mathfrak{o}) \mathfrak{o} - \phi(\gamma(\mathfrak{o})),
\end{equation}
with $\gamma(\mathfrak{o})$ the value of $\gamma$ realizing the maximum. The  G\"artner-Ellis theorem states that if $\phi(\gamma)$ is a differentiable function then $I(\mathfrak{o})=I^{**}(\mathfrak{o})$. The properties of the Legendre-Fenchel transform implies that $I^{**}(\mathfrak{o})$ is a convex function. The strict convexity of the LDF and the differentiability of the CGF are then dual conditions. Therefore we have the three following cases:
\begin{itemize}
\item[]\textbf{(Equivalence)} If $\phi(\gamma(\mathfrak{o}))$ is differentiable at $\gamma(\mathfrak{o})$, then $I(\mathfrak{o})=I^{**}(\mathfrak{o})$ from the G\"artner-Ellis theorem, and moreover $\mathfrak{o}= \partial_\gamma \phi(\gamma)$. $I(\mathfrak{o})$ is then a strictly convex function at $\mathfrak{o}$ and the equivalence holds.
\item[]\textbf{(Non equivalence)} If $\phi(\gamma(\mathfrak{o}))$ is not differentiable at $\gamma(\mathfrak{o})$ and $I(\mathfrak{o}) \neq I^{**}(\mathfrak{o})$, then $I(\mathfrak{o})$ is a non-convex function at $\mathfrak{o}$ and we have a non equivalence between the microcanonical and canonical ensembles.
\item[]\textbf{(Partial Equivalence)} If $\phi(\gamma(\mathfrak{o}))$ is not differentiable at $\gamma(\mathfrak{o})$ and $I(\mathfrak{o}) = I^{**}(\mathfrak{o})$, then $I(\mathfrak{o})$ is a convex function but not strictly convex function at $\mathfrak{o}$ and we have a partial equivalence between the microcanonical and canonical ensembles.
\end{itemize}
Hence, the non-equivalence of ensembles prevents to compute the LDF from the CGF since $ I^{**}(\mathfrak{o})$ is only the convex hull of $I(\mathfrak{o})$.  In Ref.~\cite{Chetrite2015}, the authors conjectured a connection between non equivalences and non ergodicity. We support this idea in the following by confirming that non-ergodicity is \emph{required} to observe a non equivalence of ensemble, but it is not a sufficient condition.

\subsection{Proof of ensemble equivalence for ergodic process}
\label{sec:origin-non-convexity}

The proof of ensemble equivalence amounts to determine the convexity of a LDF function. In this section, we show how the LDF for observable $ \mathcal{O}$ can be determined from the convex LDF of more general observables, and how the convexity property may be inherited in the process.

Whenever a stochastic variable depends of a set of other stochastic variables, one can infer the LDF of the former from the LDF of the latters. This operation is called a \emph{contraction}. For instance, the observable $ \mathcal{O}$ of Eq.~(\ref{eq:AT}) is a function of the empirical occupation ratio $R$ and the jump rate $C$
\begin{equation}
  \label{eq:generalObservables}
 \mathcal{O}  = \mathcal{O}(R,C) = \sum_{x} f(x) R(x) + \sum_{x,y} C(x,y) g(x,y).
\end{equation}
The occupation ratio is defined as the relative time spent in each state, say $x$ here
\begin{equation}
  \label{eq:Occupation}
  R(x) \equiv \f{1}{T} \int_0^T\delta(X_t=x) \dd t,
\end{equation}
where $\delta$ is an indicator function being $1$ if the condition is satisfied and $0$ otherwise. 
The jump rate $ C(x,y)$ counts the number of transitions between two states (from $y$ to $x$) per unit time
\begin{equation}
  \label{eq:Jumpfractions}
  C(x,y) \equiv \f{1}{T}  \sum_{0\leq t\leq T:\Delta X_t\neq 0} \delta(X_{t^+}=x) \delta(X_{t^-}=y).
\end{equation}
The occupation ratios and the jump rates have a joint probability for $R=r$ and $C=c$ that satisfies a large deviation principle
\begin{equation}
  \label{eq:LDF-R-and-C}
  P_{T}(R=r,C=c) \underset{T \to + \infty}{\simeq} e^{-T I(r,c)},
\end{equation}
defining the large deviation function $I(r,c)$, also called the level 2.5 LDF\cite{Maes2008,Barato2015,Bertini2012}. Then, the contraction 
\begin{equation}
  \label{eq:Ia}
  I(\mathfrak{o}) = \min_{r,c \text{ s.t. } \mathcal{O}(r,c) = \mathfrak{o}} I(r,c)
\end{equation}
leads to the LDF for the observable $\mathcal{O}$. In general terms, a contraction corresponds to the minimization of a multivariate LDF under a condition encoding how the stochastic variables are related \cite{Touchette2009}. 

Since ensemble equivalence for a stochastic variable relies on the convexity of the corresponding LDF, it is crucial to determine (i) whether the LDF $I(r,c)$ is convex and (ii) whether the convexity can be inherited upon contraction. The following theorem appearing in Ref.~\cite{Boyd2004} provides an answer to point (ii) when the new variable is additive:
\begin{theorem}\label{th:convexite}
Let $h(\bm{x},\bm{z})$ be a convex function and $U(\bm{z})$ an additive function, i.e. a function verifying
\begin{equation}
U(\alpha \bm{z}_1 + (1-\alpha) \bm{z}_2)=\alpha U( \bm{z}_1) +(1-\alpha) U(\bm{z}_2),
\end{equation}
then
\begin{equation}
  \label{eq:min-cont-g}
  d(\bm{y}) = \min_{\bm{x}\in \mathcal{C},\, \bm{z} \in \mathcal{C}_{\bm{y}}} h(\bm{x},\bm{z}) \text{ with }  \mathcal{C} \text{ convex, and }\mathcal{C}_{\bm{y}} = \{ \bm{z} \,|\, U(\bm{z}) = \bm{y}\}
\end{equation}
is a convex function.
\end{theorem}

\begin{proof}
  We consider $(\bm{x}^*_1,\bm{z}^*(\bm{y}_1))$ the couple of variables realizing the minimum in Eq.~(\ref{eq:min-cont-g}) for $\bm{y}_1$, and similarly $(\bm{x}^*_2,\bm{z}^*(\bm{y}_2))$ for $\bm{y}_2$. The convexity of $\mathcal{C}$ implies that $\alpha \bm{x}^*_1+(1-\alpha)\bm{x}^*_2 \in \mathcal{C}$ when $\alpha \in [0,1]$. Moreover, the additivity of $U$ implies that $\alpha \bm{z}^*(\bm{y}_1) +(1-\alpha) \bm{z}^*(\bm{y}_2) \in \mathcal{C}_{\alpha \bm{y}_1 +(1-\alpha) \bm{y}_2} $. Hence, we have
  \begin{eqnarray*}
    d(\alpha \bm{y}_1 +(1-\alpha) \bm{y}_2) 
    & = & \min_{\bm{x}\in \mathcal{C},\, \bm{z} \in \mathcal{C}_{\alpha \bm{y}_1 +(1-\alpha) \bm{y}_2} } h(\bm{x},\bm{z} ) \\
    & \leq & h(\alpha \bm{x}^*_1+(1-\alpha)\bm{x}^*_2,\alpha \bm{z}^*(\bm{y}_1) +(1-\alpha) \bm{z}^*(\bm{y}_2)), \\
    & \leq& \alpha h(\bm{x}^*_1,\bm{z}^*(\bm{y}_1)) + (1-\alpha) h(\bm{x}^*_2,\bm{z}^*(\bm{y}_2)), \\
                                  & \leq & \alpha d(\bm{y}_1) + (1-\alpha) d(\bm{y}_2),\\
  \end{eqnarray*}
where we get the third line by using the convexity of $h$, and the fourth line using our knowledge of the minimizers of Eq.~(\ref{eq:min-cont-g}) for both $\bm{y}_1$ and $\bm{y}_2$.
\end{proof}

We now address the point (i) about the convexity of the LDF $I(\bm{r},\bm{c})$. This level 2.5 LDF is explicitly known for \emph{ergodic} Markov jump processes in continuous-time~\cite{Maes2008,Barato2015,Bertini2012} 
\begin{equation}
  \label{eq:LDF2.5}
  I(\bm{r},\bm{c}) =    \sum_{x,y\neq x} \left( r(y)K(x,y)-c(x,y)+  c(x,y)\ln \f{c(x,y)}{r(y) K(x,y)} \right).
\end{equation}
It is convex since it writes as the sum of a linear part $\sum_{x,y\neq x} \left[ r(y) K (x,y) - c(x,y) \right] $ and a Kullback-Leibler divergence between $\bm{c}$ and $\bm{Kr}$ 
\begin{equation}
	D(\bm{c}||\bm{Kr}) \equiv \sum_{x,y\neq x} c(x,y)\ln \f{c(x,y)}{ K(x,y)r(y)}.
\end{equation}
This Kullback-Leibler divergence is convex as a consequence of the log-sum inequality (or Jensen's inequality)~\cite{Csiszar2004}.
Using the previous theorem with $\bm{z}=(\bm{r},\bm{c})$, the convexity of the level 2.5 LDF and the additivity of $\mathcal{O}$ for both the occupation ratio and the jump rate,
we conclude the contracted LDF $I(\mathfrak{o})$ is also convex: The ensemble equivalence holds for ergodic Markov jump processes. In particular, the ensemble equivalence holds for irreducible finite size Markov jump processes since they are always ergodic.

Alternatively, one can prove heuristically the ensemble equivalence using the differentiability of the CGF, that is the dual condition with respect to the LDF convexity. To determine the CGF differentiability, one consider the CGF as the highest eigenvalues of the titled matrix $\bm{K}_{\gamma}$ defined in Eq.(\ref{eq:tiltedmatrix}). From Perron-Fr\"obenius theorem for irreducible finite size matrices like $\bm{K}_{\gamma}$, the highest (real) eigenvalue is unique: Its multiplicity is always one independently of the value of the counting field and no crossing between the two highest eigenvalues can occur. Moreover, the components of the tilted matrix are differentiable yielding that the CGF is itself differentiable. By duality, the LDF is strictly convex. 

From this analysis, we conclude that the non-equivalence between the microcanonical and canonical ensembles based on observables like $\mathcal{O}$ may only happen when the Markov operator is reducible or of infinite dimension. Then, the LDF of the variable $\mathcal{O}$ may not be convex.

We emphasize that we made no assumption on the definition of dynamical rates, hence the system may be in or out of equilibrium. We used simple graph Markov processes, but the discussion can be transposed to include multiple channel transitions.

\section{Explicit LDF calculation for non-ergodic systems}
\label{sec:expl-comp-ldf}

In this section, we adapt the method developed in Ref.~\cite{Touchette2010} for equilibrium ensembles to compute non-convex LDF from their corresponding generating function. First, we explain how to compute a non-convex LDF from the propagator of the generating function for a simple non-ergodic system. We then apply the same method to infinite range Ising model for which the non-ergodicity emerges in the thermodynamic limit.

\subsection{A four state model with non-convex LDF of activity}
\label{sec:calc-non-conv}

Let's consider a four state system made of two subsystems with two states $(1,2)$ and $(3,4)$. The four state system evolves according to the master equation
\begin{equation}
  \label{eq:toyexampleMatrix}
  \ddf{}{t}
  \begin{pmatrix}
    p_1 \\ p_2 \\ p_3 \\ p_4
  \end{pmatrix}  =    
  \begin{pmatrix}
    -1 & 1 & 0 &0 \\
    1 & -1 & 0& 0 \\
    0 &0 & -2 & 2 \\
     0 &0&  2 & -2 \\
   \end{pmatrix} \cdot
    \begin{pmatrix}
    p_1 \\ p_2 \\ p_3 \\ p_4
  \end{pmatrix}.
\end{equation}
The Markov operator is reducible yielding to a non ergodic process. We aim at computing the system activity assuming that one cannot distinguish in which subsystem a transition occurs. The activity rate $A$ is the total number of jumps per unit time: It is given by Eq.~(\ref{eq:AT}) when $f(x)=0$ and $g(x,y)=1$ for all $(x,y)$. By definition of the Markov operator in the right hand side of Eq.~(\ref{eq:toyexampleMatrix}), the activity probability distribution of each subsystem is a Poisson distribution of mean value $T$ for the subsystem $(1,2)$ and $2T$ for the subsytem $(3,4)$, because there is $1$ and $2$ jumps per unit time in each subsystem respectively. The probability distribution of the total number of jumps is the sum of Poisson distributions weighted by the initial probability $p_{\mathrm{i},n}$ to start in state $n$. The probability of having $aT$ jumps after a time $T$ is
\begin{equation}
  \label{eq:pdfToyExample}
  P_{T}(A=a) = (p_{\mathrm{i},1}+p_{\mathrm{i},2}) \f{(T)^{aT} e^{-T}}{(aT)!} +  (p_{\mathrm{i},3}+p_{\mathrm{i},4})\f{(2T)^{aT} e^{-2T}}{(aT)!}
\end{equation}

For $T$ large enough, $P_{T}(A=a)$ will be bimodal, and the microcanonical LDF reads
\begin{equation}
  \label{eq:LDFAnalyticToyExample}
  \Imicro(a) = -\lim\limits_{T \to + \infty} \f{1}{T} \ln  P_{T}(A=a) = \left\{
    \begin{array}{ll}
        a\ln a -a +1 & \mbox{if } a < 1/\ln 2\\
        a \ln \f{a}{2} -a +2 & \mbox{if } a \geq 1/ \ln 2
    \end{array}
\right. ,
\end{equation}
where we have used Stirling formula and chosen the minimum of the two LDF corresponding to each Poisson distribution. The above microcanonical LDF of activity is shown in Fig.~\ref{fig:SimpleExample}b. It is crucial to note that the initial condition plays a fundamental role here: if the system never starts in states $1$ or $2$, i.e. $p_{\mathrm{i},1}=p_{\mathrm{i},2}=0$, the LDF will include the branch corresponding to the second line of Eq.~(\ref{eq:LDFAnalyticToyExample}) only. The non-ergodicity impacts the long time statistics of the dynamical observables through the choice of initial conditions.

For the four state model, it is straightforward to determine the probability density function of activity, but for other systems or observables this task may be more challenging: one must often compute the CGF instead. Let's derive the result of Eq.~(\ref{eq:LDFAnalyticToyExample}) in this way using the propagator for the generating function of the activity defined by $ G(x_{\mathrm{f}},x_{\mathrm{i}},\gamma)=\langle e^{T\gamma A} \rangle_{x_{\mathrm{f}},x_{\mathrm{i}}} $, where the subscripts $x_{\mathrm{i}}$ and $x_{\mathrm{f}}$ denote respectively the initial and final states of the trajectories appearing in the average. Using standard approach \cite{VanKampen2007,Touchette2009}, the propagator is obtained from the biased operator
\begin{equation}
  \label{eq:biasedOpToyExemple}
  \bm{K}_{\gamma} =
    \begin{pmatrix}
    -1 & e^\gamma & 0 &0 \\
    e^\gamma & -1 & 0& 0 \\
    0 &0 & -2 & 2e^\gamma \\
     0 &0&  2e^\gamma & -2 \\
   \end{pmatrix}
\end{equation}
as the components of its matrix exponential $\left[\exp ( T \bm{K}_{\gamma}) \right](x_{\mathrm{f}},x_{\mathrm{i}}) = G(x_{\mathrm{f}},x_{\mathrm{i}},\gamma)$. Using eigenvalues and eigenvectors decomposition of $\bm{K}_{\gamma}$, we make explicit the matrix exponential
\begin{multline}
  \label{eq:propagatorToyExemple}
  \exp \left[ T \bm{K}_{\gamma} \right] =
  e^{T(-1+e^\gamma)}  \begin{pmatrix}1/2\\1/2\\0\\0 \end{pmatrix}\cdot \begin{pmatrix}1&1&0&0 \end{pmatrix}   
  +e^{2T(-1+e^\gamma)} \begin{pmatrix}0\\0\\1/2\\1/2 \end{pmatrix}\cdot \begin{pmatrix}0&0&1&1 \end{pmatrix} \\
  +  e^{-T(1+e^\gamma)}  \begin{pmatrix}1/2\\-1/2\\0\\0 \end{pmatrix}\cdot \begin{pmatrix}1&-1&0&0 \end{pmatrix}
  + e^{-2T(1+e^\gamma)}  \begin{pmatrix}0\\0\\1/2\\-1/2 \end{pmatrix}\cdot \begin{pmatrix}0&0&1&-1 \end{pmatrix},
\end{multline}
The orthogonal basis of eigenvectors has normalized right eigenvectors, and the scalar products of the left and right eigenvectors associated to the same eigenvalue are all equal to $1$. We notice that the eigenvectors separate in two sets whose supports are disjoints and correspond to each subsystem respectively. We remark also that the above propagator should be norm conserving when $\gamma = 0$, but the two terms in the second line of Eq.~(\ref{eq:propagatorToyExemple}) do not fulfill this requirement. We do not consider then as physical and keep only the two first terms in Eq.~(\ref{eq:propagatorToyExemple}).
Then, the activity LDF is recovered from this propagator by first summing over the initial and final states 
\begin{equation} \label{eq:GenFunction}
\langle e^{T\gamma A} \rangle = \sum_{x_{\mathrm{f}},x_{\mathrm{i}}}  G(x_{\mathrm{f}},x_{\mathrm{i}},\gamma) p_{\mathrm{i}}(x_{\mathrm{i}}) ,
\end{equation}
second, applying an inverse Laplace transform, and finally by taking the limit $T \rightarrow \infty$. For ergodic systems, this procedure leads to the same LDF whatever the order of these operations. On the contrary, the order matters for non-ergodic systems. 

For our $4$ state model, the generating function writes
 \begin{equation}
   \label{eq:GFToyExample}
    \langle e^{T\gamma A} \rangle = (p_{\mathrm{i},1}+p_{\mathrm{i},2}) e^{T(-1+e^\gamma)}+ (p_{\mathrm{i},3}+p_{\mathrm{i},4}) e^{2T(-1+e^\gamma)}.
  \end{equation}
Its inverse Laplace transform yields the probability density function of activity given in Eq.~(\ref{eq:pdfToyExample}). However, first computing the long time limit of the generating function leads to the CGF
  \begin{equation}
    \label{eq:CGFToyExemple}
    \phi(\gamma) = \lim\limits_{T \to + \infty} \f{1}{T} \ln \sum_{x_{\mathrm{f}},x_{\mathrm{i}}} G(x_{\mathrm{f}},x_{\mathrm{i}},\gamma)p_{\mathrm{i}}(x_{\mathrm{i}})= \left\{
    \begin{array}{ll}
       e^{\gamma}-1 & \mbox{if } \gamma < 0,\\
        2( e^{\gamma}-1) & \mbox{if } \gamma \geq 0
    \end{array}
\right. ,
\end{equation}
as long as $p_{\mathrm{i}}(x) > 0$ for all x.
Noticing that the limit $T \to + \infty $ enables to use a saddle-point method to approximate the inverse Laplace transform into a Legendre-Fenchel transform, the asymptotic probability of activity follows from its corresponding LDF 
  \begin{equation}
    \label{eq:LDFConvexToyExemple}
      \Icano(a) = \max_{\gamma} \left[ a\gamma- \phi(\gamma)  \right] =   \left\{
    \begin{array}{ll}
      a\ln a -a +1 & \mbox{if } a < 1\\
      0 & \mbox{if } a \in [1,2] \\
        a \ln \f{a}{2} -a +2 & \mbox{if } a > 2
    \end{array}
\right. ,
\end{equation}
that is not the one of Eq.~(\ref{eq:LDFAnalyticToyExample}).
The above LDF is convex because a Legendre-Fenchel transform only yields convex functions by definition, while the LDF of Eq.~(\ref{eq:LDFAnalyticToyExample}) is not convex.

Alternatively, one may obtain the microcanonical LDF by taking the long time limit on the propagator of the generating function of Eq.~(\ref{eq:propagatorToyExemple}), and not on the generating function itself, yielding
\begin{equation}
  \label{eq:ToyExampleLargeTimepropGF}
  \phi_{x_{\mathrm{i}}}(\gamma)  = \lim\limits_{t \to + \infty} \f{1}{t} \ln G(x_{\mathrm{f}},x_{\mathrm{i}},\gamma,t)  = \left\{
    \begin{array}{ll}
		(e^\gamma-1) 	& \text{ if } x_{\mathrm{i}}=1,2 \\
		2(e^\gamma-1)	& \text{ if } x_{\mathrm{i}}=3,4
    \end{array}
\right. .
\end{equation}

\begin{figure}
  \centering
  \includegraphics[width=\columnwidth]{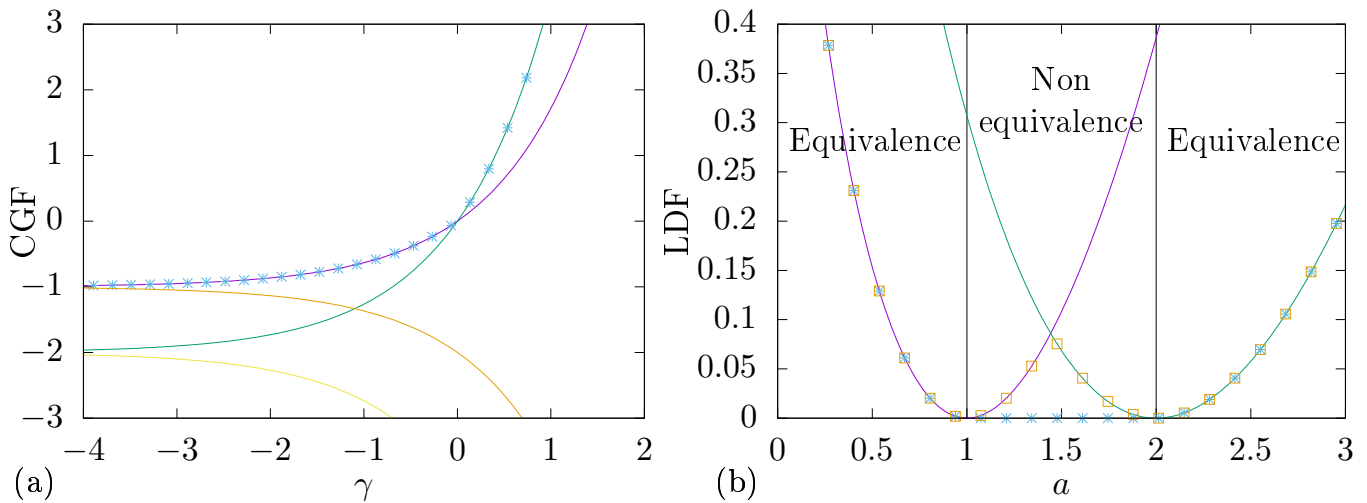}
  \caption[Simple Example]{(a) Four eigenvalues of the biased operator (\ref{eq:biasedOpToyExemple}) (solid line) and corresponding CGF (cross).  (b) Partial LDFs (solid lines) corresponding to the Legendre transform of the two highest eigenvalues, canonical LDF (cross) and microcanonical LDF (squares). 
  }
  \label{fig:SimpleExample}
\end{figure}
The Legendre-Fenchel conjugate of these two branches associated to different initial states are precisely the two branches of the microcanonical LDF in Eq.~(\ref{eq:LDFAnalyticToyExample}):
\begin{equation}
	I_{x_{\mathrm{i}}}(a) = \max_{\gamma} \left[ a\gamma- \phi_{x_{\mathrm{i}}}(\gamma)  \right] = \left\{
    \begin{array}{ll}
		 	a\ln a -a +1 & \text{ if } x_{\mathrm{i}}=1,2 \\
			a \ln \f{a}{2} -a +2 & \text{ if } x_{\mathrm{i}}=3,4
    \end{array}
\right. \label{eq:LDFwithIC}
\end{equation}
The summation over initial and final conditions is now carried out using an asymptotic approximation, which can be written heuristically as
\begin{equation}
	e^{-T \Imicro(a)} \simeq \sum_{x_{\mathrm{i}}} p_{\mathrm{i}}(x_{\mathrm{i}}) e^{-T I_{x_{\mathrm{i}}}(a)} \simeq \exp \left( -T  \min_{x_{\mathrm{i}}} I_{x_{\mathrm{i}}}(a) \right)
\end{equation}
explaining why a minimum on the branches of Eq.~(\ref{eq:LDFwithIC}) appears in the final microcanonical LDF of Eq.~(\ref{eq:LDFAnalyticToyExample}). To summarize, from the Legendre-Fenchel conjugate of all the eigenvalues appearing in the biased operator of Eq.~(\ref{eq:biasedOpToyExemple}), we determined the partial LDFs. The minimum among these partial LDFs produces the microcanonical LDF. We illustrate this procedure in Fig.~\ref{fig:SimpleExample}. From this figure, we conclude on the ensemble equivalence for this model: it holds for $a \in [0,1[ \,\bigcup\, ]2,+\infty[ $, but the equivalence is partial at $a={1,2}$ and there is no equivalence for $a\in ]1,2[$.

Remark that when the mean activities of each subsystem are the same, the two branches of the LDF merge toward a convex shape. Therefore, in this case non-ergodicity is not sufficient to observe ensemble non-equivalence associated to non-convexity of the LDF.

\subsection{General framework}
\label{sec:generalSetup}

In the above example, we have seen that one can determine the asymptotic fluctuations of a physical observable by switching from one dynamical ensemble to another as long as the LDFs are piecewise-convex.
In this section, we develop this approach for the more general framework of section \ref{sec:PathEnsembles}. In practice, this amounts to express the inverse Laplace transform as a Legendre-Fenchel transform, using the saddle point method and taking care of the initial condition appropriately.

By definition, for $b \in \mathbb{R} $, the microcanonical LDF for $\mathcal{O}$ writes 
\begin{equation}
  \label{eq:OrderOperationNonConvex}
  \Imicro(\mathfrak{o}) \equiv \lim\limits_{T \to + \infty} \f{-1}{T} \ln  \int_{\bromwich-i\infty}^{\bromwich+i\infty} \dd \gamma \,  e^{-T\gamma \mathfrak{o}} \sum_{x_{\mathrm{f}},x_{\mathrm{i}}} G(x_{\mathrm{f}},x_{\mathrm{i}},\gamma) p_{\mathrm{i}}(x_{\mathrm{i}}),
\end{equation}
in term of the initial probability $p_{\mathrm{i}}(x_{\mathrm{i}})$ and of the propagator $ G(x_{\mathrm{f}},x_{\mathrm{i}},\gamma) = \langle e^{T\gamma \mathcal{O}} \rangle_{x_{\mathrm{f}},x_{\mathrm{i}}} $ that generates the moments of $\mathcal{O}$ under given initial and final conditions. The argument of the logarithm in Eq.~\eqref{eq:OrderOperationNonConvex} is exactly the probability distribution function of $\mathcal{O}$. Since solely the most probable events contribute to the LDF, we can focus on the initial conditions leading to the minimal value of the LDF (as seen in section \ref{sec:calc-non-conv}):
\begin{equation}
  \label{eq:OrderOperationNonConvex1}
  \Imicro(\mathfrak{o}) = \min\limits_{x_{\mathrm{f}},x_{\mathrm{i}}} \lim\limits_{T \to + \infty} \f{-1}{T} \ln  \int_{\bromwich-i\infty}^{\bromwich+i\infty} \dd \gamma \,  e^{-T\gamma \mathfrak{o}} G(x_{\mathrm{f}},x_{\mathrm{i}},\gamma) p_{\mathrm{i}}(x_{\mathrm{i}}).
\end{equation}
Finally, the complex integral for the inverse Laplace conjugate follows from the saddle point method:
\begin{equation}
  \label{eq:OrderOperationNonConvex2}
  \Imicro(\mathfrak{o}) = \min\limits_{x_{\mathrm{f}},x_{\mathrm{i}}} \max\limits_{\gamma} \left( \gamma \mathfrak{o} - \lim\limits_{T \to + \infty} \f{1}{T} \ln G(x_{\mathrm{f}},x_{\mathrm{i}},\gamma)p_{\mathrm{i}}(x_{\mathrm{i}}) \right).
\end{equation}

Alternatively, the convex hull of $\Imicro(\mathfrak{o})$ is the Legendre-Fenchel conjugate of the CGF for the observable $\mathcal{O}$, namely
\begin{equation}
  \label{eq:OrderOperationConvexHull}
  \Icano(\mathfrak{o}) \equiv  \max\limits_{\gamma} \left( \gamma \mathfrak{o} - \lim\limits_{T \to + \infty} \f{1}{T} \ln \sum_{x_{\mathrm{f}},x_{\mathrm{i}}} G(x_{\mathrm{f}},x_{\mathrm{i}},\gamma) p_{\mathrm{i}}(x_{\mathrm{i}})\right), 
\end{equation}
As seen in section \ref{sec:calc-non-conv}, the propagator $G$ may have an eigendecomposition for which the terms contributing to the CGF in the limit $T \to \infty$ depend on the initial or final condition. Hence, a maximum on $x_{\mathrm{f}}$ and $ x_{\mathrm{i}}$ must appear
\begin{equation}
	\Icano(\mathfrak{o}) =  \max\limits_{\gamma}  \left(  \gamma \mathfrak{o} -  \max\limits_{x_{\mathrm{f}},x_{\mathrm{i}}} \lim\limits_{T \to + \infty} \f{1}{T} \ln   G(x_{\mathrm{f}},x_{\mathrm{i}},\gamma)p_{\mathrm{i}}(x_{\mathrm{i}}) \right),
\end{equation}
to select the dominant asymptotic behavior in the limit $T\to\infty$. In view of comparing with Eq.~(\ref{eq:OrderOperationNonConvex2}), the maximization is modified into a minimization through the commutation with  the minus sign:
\begin{equation}
	\Icano(\mathfrak{o}) =  \max\limits_{\gamma}  \min\limits_{x_{\mathrm{f}},x_{\mathrm{i}}}\left(  \gamma \mathfrak{o} -  \lim\limits_{T \to + \infty} \f{1}{T} \ln   G(x_{\mathrm{f}},x_{\mathrm{i}},\gamma)p_{\mathrm{i}}(x_{\mathrm{i}}) \right).	
\end{equation}
In the end, the difference between the LDF $\Imicro(\mathfrak{o})$ and its convex hull $\Icano(\mathfrak{o})$ comes from the non commutation of $ \max_{\gamma}$ and $\min_{x_{\mathrm{f}},x_{\mathrm{i}}}$ as a consequence of the dependence on the initial conditions, i.e. of the non-ergodicity. Of course, if the LDF were convex, the ordering of these extremization would not matter.

The eigendecomposition of the propagator $G$ involves the eigenvalues of the biased matrix of Eq.~(\ref{eq:tiltedmatrix}). For non-ergodic systems or when the state space is infinite, the assumption of the Perron-Frobenius theorem does not hold and several eigenvalues may cross each other. To compute exactly the microcanonical LDF, one needs all the branches corresponding to each eigenvalue that becomes the highest eigenvalue for at least one $\gamma$. Only those branches matters, and
other eigenvalues will not contribute, so as can be understood from the following argument: Be two eigenvalues $\phi_1(\gamma)$ and $\phi_2(\gamma)$ such that $\phi_1(\gamma)>\phi_2(\gamma)$ for all $\gamma $. Since $\gamma \mathfrak{o} -\phi_1(\gamma) < \gamma \mathfrak{o} - \phi_2(\gamma)$, we have
\begin{equation}
  \label{eq:elimiationpetiteEV}
  I_1(\mathfrak{o}) = \min\limits_{\gamma} \left\{ \gamma \mathfrak{o}-\phi_1(\gamma) \right\} < I_2(\mathfrak{o}) = \min\limits_{\gamma} \left\{ \gamma \mathfrak{o}-\phi_2(\gamma) \right\}.
\end{equation}
The last minimization in Eq.~(\ref{eq:OrderOperationNonConvex2}) on the partial LDFs will withdraw the contribution coming from the eigenvalue $\phi_{2}(\gamma)$ if it is smaller that $\phi_{1}(\gamma)$.

Physically speaking, we study a rare event in a system that has several independent subparts. We assume that the initial probability cannot be zero in all states of a subsytem, otherwise this subsystem shall be ignored. Each subpart of the system has its own probability to realize the rare event at stake. The subpart for which the event is the most likely will determine the event probability. This will be so if the rare event corresponds to a fluctuation of a time average quantity on a sufficiently long time so as to neglect the role of the initial state probability.

Mathematically speaking, when dealing with reducible biased operators whose highest eigenvalue is the CGF of interest, we must divide the operator into irreducible sub-operators for which holds the ensemble equivalence. For every sub-operators we proceed normally using the ensemble equivalence to determine the partial LDFs of the chosen observable from the Legendre-Fenchel transform of the highest eigenvalue of the sub-operator. The final LDF for the total system is then given by the minimum over all partial LDFs. This explains why the final LDF is piecewise convex.


\subsection{Mean-field Ising model}
\label{sec:mean-field-ising}

Using the results of sec.~\ref{sec:generalSetup}, we study the activity and magnetization of an infinite range Ising model. This model is ergodic when considering a finite number of spins, but breaks ergodicity in the thermodynamics limit. In the following, we first introduce the model and its mean field (MF) treatment. Second, we provide the propagator of the generating function for magnetization and activity, and next use it to determine the CGF and both the canonical and microcanonical LDFs.
 
\subsubsection{Model description and thermodynamics limit}
\label{sec:modelDef}

We consider the fully connected Ising model made of $N$ interacting spins $\{s\}=(s_{1},\dots,s_{N})$. Each spin $s_{i}$ can hop between states $+1$ and $-1$ by exchanging  heat with a thermostat at inverse temperature $\beta$. The interaction energy between two spins is $V /N$ when the spins are not aligned and vanishes for parallel spins. The interaction energy is independent of the distance between spins. Beside the spin-spin interaction, each spin has a potential energy $-s_{i}E$ due to the presence of an external magnetic field. We introduce the free energy $F(n) \equiv U(n)-S^\IN(n)/\beta$ of the mesostate $n= \sum_{i=1}^N s_i$ in term of the total energy given (up to a constant) by
\begin{equation}
  \label{eq:energy}
  U(n) \equiv- \f{V}{N} \sum_{1\leq i\leq j \leq N}  s_is_j-E \sum_{1\leq i \leq N} s_i= -n^2 \f{V}{2N}-n E
\end{equation}
and of the internal entropy 
\begin{equation}
S^\IN(n) \equiv \ln N!/[\left(\f{N+n}{2}\right)!\left(\f{N-n}{2}\right)!].
\end{equation} 
We chose the transition rate from $n$ to $n+2\epsilon$ (with $\epsilon = \pm 1$) to be
\begin{equation}
  \label{eq:ratesfunction}
  K(n+2\epsilon,n) \equiv \Gamma \left ( N-\epsilon  n \right ) e^{ \f{\beta}{2} \left( (2\epsilon n +2)V/N + 2\epsilon E \right) }.
\end{equation} 
In the following we chose $\Gamma = 1$ and $\beta = 1$ to set the time and energy scales respectively. The system is in thermal equilibrium and the transition rates satisfy the detail balance equation
\begin{equation}
\ln \frac{K(n+2\epsilon,n)}{K(n, n+2\epsilon)} = -\beta (F({n+2\epsilon})-F({n})) \label{eq:DB},
\end{equation}
The probability of state $n$ at time $t$, denoted $p_n = p_n(t)$ evolves according to the master equation
\begin{equation}
 \dot p_n = \sum_{\epsilon=\pm 1} \left[  K(n,n+2\epsilon) p_{n+2\epsilon}  -  K(n+2\epsilon,n) p_{n} \right],\label{eq:MasterEquation}
\end{equation}
where $\dot p_n$ is the time derivative of $p_n$. 
The stationnary probability $\pstat(n)$ is the equilibrium probability
 \begin{equation}
   \label{eq:probastat}
   \pstat(n) \equiv \f{1}{Z\e{eq}} e^{-F(n)} \quad \text{with} \quad Z\e{eq} \equiv \sum_{n=0}^{N} e^{-F(n)}.
 \end{equation}

The time-averaged stochastic magnetization is obtained from Eq.~(\ref{eq:AT}) by chosing the state dependent function $f(n) = n/N$ and a vanishing function $g(n,n')=0$ for all $(n,n')$:
\begin{equation}
	M \equiv \frac{1}{NT} \int_{0}^{T} dt\, n(t),
\end{equation}
where $n(t)$ is the mesostate at time $t$.
We denote by $m$ some real value that can be achieved by the stochastic variable $M$. The mean magnetization in the stationary state writes
\begin{equation}
  \label{eq:mean-density}
  \langle M \rangle = \f{1}{N} \sum_{n=0}^{N} n \pstat (n).
\end{equation}
The activity is obtained from Eq.~(\ref{eq:AT}) by chosing $g(n,n')=1/N$ and $f(n) = 0$ for all $(n,n')$:
\begin{equation}
  \label{eq:AT}
A \equiv \f{1}{NT} \sum_{0\leq t\leq T:\Delta n_{t}\neq 0} 1
\end{equation}
The activity rate is thus a time-symmetric observable : the number of spin flips per unit time and per spin are identical in a trajectory and its time reversal. We denote by $a$ some real value that can be achieved by  the stochastic variable $A$. The mean activity in the stationary state writes
\begin{equation}
  \label{eq:mean-activity}
    \langle A \rangle = \f{1}{N} \sum_{n=0}^{N} \sum_{\epsilon=\pm 1} \pstat (n)  K_{n+2\epsilon,n} .
 \end{equation}

In the thermodynamic limit, when taking the continuous limit for the mesostate $x\equiv n/N \in [-1,1]$, the system energy changes by
\begin{equation}
\lim_{N \rightarrow \infty} \left[ U(n+2\epsilon)-U(n) \right] = 2\epsilon (V x+E)
\end{equation}
for a transition from $n$ to $n+2\epsilon$. Accordingly, the transition rates are in the same limit
\begin{equation}
  \label{eq:TransitionratesMF}
J_{\epsilon}(x) \equiv \lim\limits_{N\to+\infty} \f{K_{xN+2\epsilon,xN}}{N} = \left ( 1 -\epsilon x \right ) e^{ \epsilon \left[  V x+ E \right] }.
\end{equation}
In this case, the master equation Eq.~(\ref{eq:MasterEquation}) can be transformed into an evolution equation for $x$ in the mean field (MF) approximation
\begin{equation}
\langle \dot x \rangle =  \sum_{\epsilon =\pm 1} \epsilon J_{\epsilon}\left(\langle x \rangle \right).
\label{eq:MFequation}
\end{equation}
The steady state solution of this equation is the mean field magnetization $x = \obsm^\MF$ verifying
\begin{equation}
J_{-}\left (\obsm^\MF \right ) = J_{+}\left (\obsm^\MF\right). \label{eq:SteadyMFME}
\end{equation}
Using Eq.~(\ref{eq:TransitionratesMF}), the previous equation is equivalent to the transcendental equation
\begin{equation}
  \label{eq:MF-mag}
  \obsm^\MF = \tanh \left(V\obsm^\MF+E\right).
\end{equation}
The MF activity follows from the mean-field magnetization:
\begin{equation}
  \label{eq:MF-activity}
  a^\MF = J_{-}\left (\obsm^\MF \right )+ J_{+}\left (\obsm^\MF\right)= \left[1- \obsm^\MF\right]e^{ V\obsm^\MF+E}+\left[1+ \obsm^\MF\right]e^{ -V\obsm^\MF-E}.
\end{equation}
The MF magnetization and activity are shown in the bifurcation diagram of Fig.\ref{fig:MF}. At a critical value of the interaction energy, three MF magnetizations appear instead of a unique one, due to the well studied ferromagnetic transition in the Ising model. This bifurcation also affects the system activity as shown on the Fig.\ref{fig:MF}b and as expected from  Eq.~(\ref{eq:MF-activity}) since $a^\MF$ is a function of $\obsm^\MF$. We remark that for $E=0$ the system activity is an even function of the magnetization and hence the bifurcation diagram for activity has two branches only. We also notice that the MF activity is higher for the branch that does not break the system symmetry.
\begin{figure}[h]
\centering
\includegraphics[width=\columnwidth]{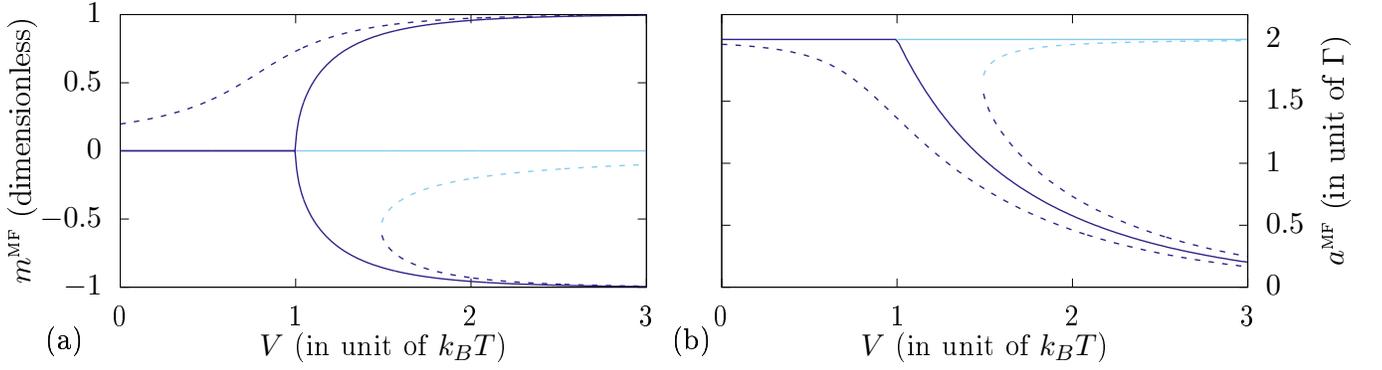}
\caption[Mean-field]{\label{fig:MF} 
(a) Stable (dark blue) and unstable (light blue) mean field steady state densities $m^\MF$ versus interaction energy $V$. 
(b)  Stable (dark blue) and unstable (light blue) mean field steady state activity $a^\MF$ versus interaction energy $V$. $E=0$ for solid lines and $E=0.2$ for dashed lines.
}
\end{figure}

\subsubsection{Propagator of the generating function for magnetization and activity}

Like in the four state model of section~\ref{sec:calc-non-conv}, we look for the propagator of the generating function for the activity and magnetization:
\begin{equation}
  \label{eq:propagatorDef}
 G(x_{\mathrm{f}},x_{\mathrm{i}},\kappa,\gamma) = \langle e^{N T (\kappa \obsmag+\gamma a)} \rangle_{x_{\mathrm{f}},x_{\mathrm{i}}},
\end{equation}
where as before the subscripts indicate a conditioning on the initial and final states, respectively $x_{\mathrm{i}} = n_{\mathrm{i}}/N$ and $x_{\mathrm{f}}= n_{\mathrm{f}}/N$. From a path integral approach, see Appendices~\ref{app:derivationPropagator} and~\ref{app:WKBapprox}, an asymptotic expression of the propagator reads
\begin{equation}
  \label{eq:propagatorGFclass}
G(x_{\mathrm{f}},x_{\mathrm{i}},\kappa,\gamma) \simeq_{N\to +\infty}  \exp \left( N \int_{0}^{T} \dd  t' \mathcal{L}(x_{t'},\dot{x}_{t'},\kappa,\gamma) \right ),
\end{equation}
where $\mathcal{L}$ is the Lagrangian
\begin{equation}\label{eq:Lagrangien}
\mathcal{L}(x,\dot{x},\kappa,\gamma)=\f{\dot{x}}{2}\ln\left( \f{- \dot x + \sqrt{\dot{x}^2+\varphi(x,\gamma)}}{2 J_-(x) e^{\gamma}} \right)-\sum_{\epsilon=\pm1} J_\epsilon(x)+\sqrt{\dot{x}^2+\varphi(x,\gamma)}+\kappa x,
\end{equation}
with
\begin{equation}
  \label{eq:varphi}
\varphi(x,\gamma)=4\prod\limits_{\epsilon=\pm1}J_\epsilon(x)e^{\gamma} = 4(1+x)(1-x)e^{2\gamma }.
\end{equation}
The propagator of Eq.~(\ref{eq:propagatorGFclass}) is almost explicit: the path $[x]$ starting in $x_{\mathrm{i}}$ and ending in $x_{\mathrm{f}}$ must be determined using the Euler-Lagrange equation
\begin{equation}
  \label{eq:Euler-Lagrange}
\Dp{\mathcal{L}}{x}=\ddf{}{t}\left( \Dp{\mathcal{L}}{\dot x}\right).
\end{equation}
The propagator of the generating function is so determined by the initial and final conditions.

As explained in sec.~\ref{sec:generalSetup}, we now vary initial and final conditions to obtain the dominant contributions to the generating function. As the large size and long time limit of the propagator is evaluated with large size limit first and then long time limit, the system is non-ergodic. Due to this non-ergodicity, the dominant contributions are obtained from stationary trajectories, \textit{i.e.} element of the propagator $G(\xetoile,\xetoile,\kappa,\gamma)$ such that 
\begin{equation}
\Dp{\mathcal{L}}{x} (\xetoile,0,\kappa,\gamma)=0.\label{eq:stat_condition}
\end{equation}
and we denote $\xetoile(\kappa,\gamma)$ the various solutions of this equation. These constant trajectories are the only ones that will dominate for at least one value of $(\kappa,\gamma)$. Following sec.~\ref{sec:generalSetup}, we can restrain the extremization over initial and final conditions to these stationary trajectories.

We provide an heuristic argument for the choice of stationary trajectories: Dominant trajectories correspond to long-time behavior of the tilted system. We expect the tilted system to be in a stationary state at long-time, as we study an equilibrium system. Neglecting the boundary terms due to long time limit, we restrain to trajectories starting and ending into a stationary state. The system being non-ergodic, the state space is divided into various subpart each associated to different $\xetoile$ and we forbid trajectories that start and end into different subpart of the state space. This approach is confirmed by numerical computation of the CGF in the next section.

\subsubsection{CGF of magnetization and activity}

We now use the propagator of Eq.~(\ref{eq:propagatorGFclass}) to derive the CGF of magnetization and activity. This CGF proceeds from the leading elements of the propagator of the generating function
\begin{eqnarray}
  \label{eq:CGFfrompropagatorGF}
  \phi(\kappa,\gamma) &=& \lim\limits_{N,T \to \infty} \f{1}{NT} \ln \langle e^{NT \kappa \obsmag + NT\gamma a} \rangle, \nonumber\\
                      &=& \max\limits_{x_{\mathrm{f}},x_{\mathrm{i}}} \lim\limits_{N,T \to \infty} \f{1}{NT} \ln  G(x_{\mathrm{f}},x_{\mathrm{i}},\kappa,\gamma) p_{\mathrm{i}}(x_{\mathrm{i}}).
\end{eqnarray}
As explained above, we can focus on stationary trajectories that solve Eq.~\eqref{eq:stat_condition}.
Then, solving for $x$ amounts to find the extrema of $\mathcal{L}(x,0,\kappa,\gamma)$.
Assuming that $p_{\mathrm{i}}(x)>0$ for all $x$, one can safely drop the initial probability $p_{\mathrm{i}}$ in Eq.~(\ref{eq:CGFfrompropagatorGF}) ending with
\begin{eqnarray}\label{eq:CGFcanonique}
\phi(\kappa,\gamma) &=& \max_{x} \mathcal{L}(x,0,\kappa,\gamma), \\ &=&   \sqrt{\varphi(\xetoile(\kappa,\gamma),\gamma) } 
-\sum_{\epsilon=\pm1,\nu=1,2} J_\epsilon^\nu(\xetoile(\kappa,\gamma))+\kappa \xetoile(\kappa,\gamma).
\end{eqnarray}
\begin{figure*}[h]
\includegraphics[width=\columnwidth]{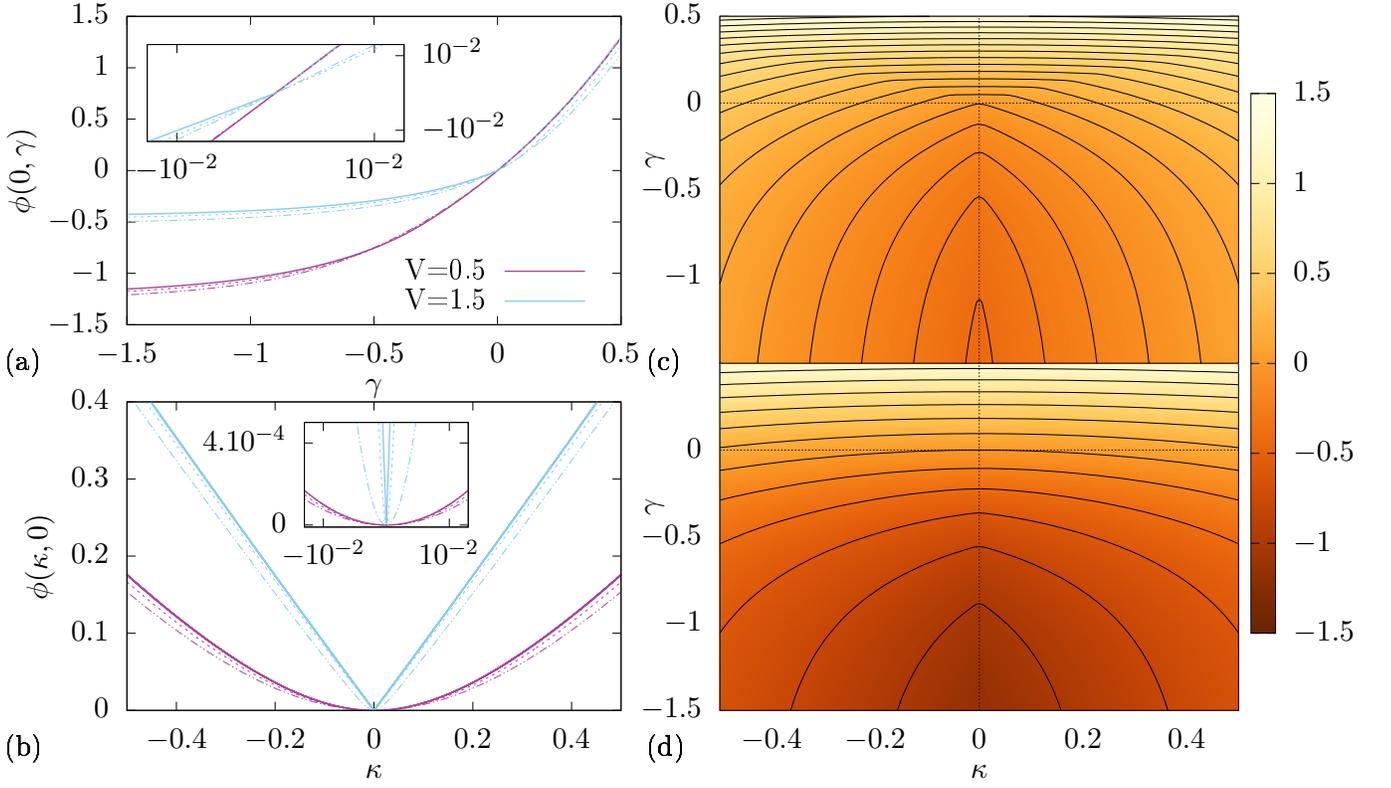}
\caption[CGF]{
Cross-sections of the CGF $\phi(\kappa,\gamma)$ along (a) $\kappa=0$ and (b) $\gamma=0$ (solid lines), and corresponding CGFs for the finite size system with $N=10$ (dot-dashed lines) and $N=25$ particles (dashed lines). 
Inserts: Zoom on the non-differentiable point of the light blue line for which $V=1.5$. 
(c) CGF $\phi(\kappa,\gamma)$ and level lines for $V=1.5$ (d) CGF $\phi(\kappa,\gamma)$ and level lines for $V=0.5$. For all figures E= 0. 
 \label{fig:CGF}
}
\end{figure*}
Unfortunately, the determination of $\xetoile(\kappa,\gamma)$ involves a transcendental equation. We solved this equation numerically to provide the CGFs before and after the bifurcation in Fig~\ref{fig:CGF}d and \ref{fig:CGF}c respectively. Cross-sections of the CGF in the plane $\kappa=0$ and $\gamma=0$ are shown in Fig~\ref{fig:CGF}a-b. After the bifurcation for $V=1.5$, the CGF is clearly not differentiable. The left and right partial derivatives at $(\kappa,\gamma)=(0,0)$ lead to different mean magnetization and activity in agreement with the bifurcation diagram of Fig~\ref{fig:MF} of the mean-field framework. We notice on Fig~\ref{fig:CGF}d that before the transition the CGF has a non differentiability not located at the origin of the $(\gamma,\kappa)$ plane. Hence the mean magnetization and activity are unique, but their fluctuations are impacted by the phase transition. We confirm our results by computing numerically the CGF as the highest eigenvalue of the biased matrix 
\begin{equation}
  \label{eq:tiltedmatrix}
  \bm{K}_{\gamma,\kappa}(x,y)= \left\{\begin{array}{ll} k(x,y)e^{\gamma/N}  & \text{ if } x \neq y \\k(x,x)+\kappa x/N & \text{ if } x=y \end{array} \right. ,
\end{equation}
for systems with $N=10$ and $N=25$ spins.
As explained in Sec~\ref{sec:origin-non-convexity}, the CGFs of finite size systems are everywhere differentiable. Then, they only approach gradually the non-differentiable CGF when $N\to \infty$.

\subsubsection{Canonical LDF of magnetization and activity}
\label{sec:conv-ldf-magn}

The canonical LDF $\Icano (\obsmag,a)$ shown in Fig.~\ref{fig:ConvexLDF}c-d is the Legendre-Fenchel conjugate of the CGF
\begin{equation}\label{eq:LegendreTransform}
\Icano(\obsmag,a) = \max\limits_{\kappa,\gamma} \left\{ \kappa \obsmag+ \gamma a -\phi(\kappa,\gamma) \right\}.
\end{equation}
At low interaction energy $V$, we observe a smooth function whose unique minimum is given by the mean-field solution of Eqs.~(\ref{eq:MF-mag}-\ref{eq:MF-activity}). However at higher interaction energy, a plateau appears in the LDF between the MF solutions. This plateau, in association with the non-differentiability of the CGF, indicates a phase transition. As emphasized before, we always have ergodicity within the canonical ensemble, and this plateau is the result of a ``temporal coexistence'' of MF states: the system spends most of its time into the various MF states leading to a time averaged magnetization and activity belonging to the convex area defined by the MF solutions, here a triangle. 

The contracted LDF for activity $\Icano(a)$ and magnetization $\Icano(\obsmag)$ defined by
\begin{eqnarray}
  \label{eq:contractIconvex}
  \Icano(\obsmag) \equiv \min\limits_{a} \Icano(\obsmag,a) =   \max\limits_{\kappa} \left\{ \kappa \obsmag -\phi(\kappa,0) \right\}\\
   \Icano(a) \equiv \min\limits_{\obsmag} \Icano(\obsmag,a) = \max\limits_{\gamma} \left\{ \gamma a -\phi(0,\gamma) \right\},
\end{eqnarray}
are shown in Fig.~\ref{fig:ConvexLDF}a-b, together with the LDFs for the finite size systems obtained from the Legendre transform of their corresponding CGFs in Fig.~\ref{fig:CGF}. For $V=1.5$, the latter LDFs converge towards the plateau with a speed that is lower for the activity LDF than for the magnetization LDF in agreement with the fact that the plateau for the LDF of activity lies between a stable MF solution and an unstable MF solution, while the plateau for the magnetization lies between two stable solutions.
\begin{figure*}[h]
\includegraphics[width=\columnwidth]{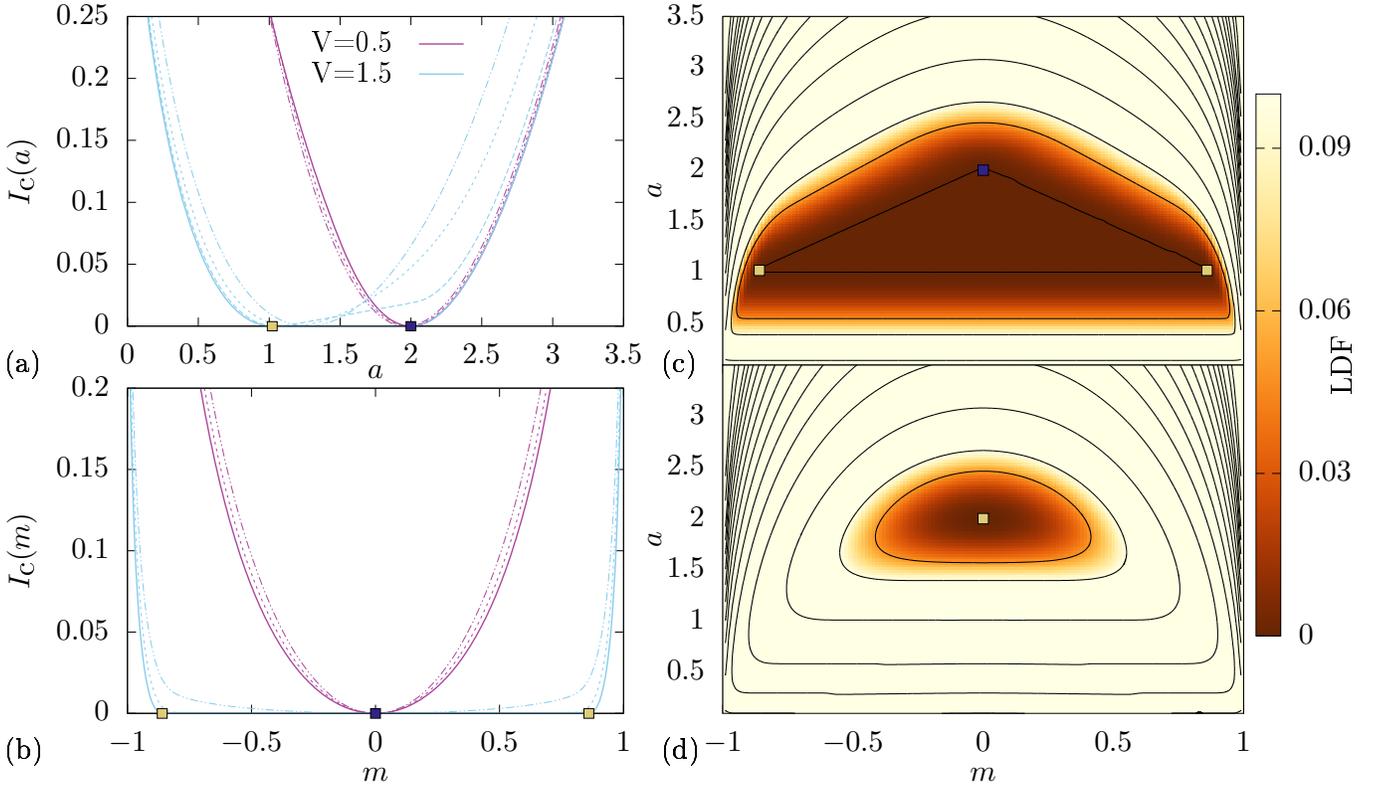}
\caption[Convex LDF]{
(a) Canonical LDF $ \Icano(a)$ (solid lines) and finite size LDF for $N=10$ (dot-dashed lines), $N=25$ (dashed lines) and N=$100$  (long-dashed lines). 
(b) Canonical LDF $\Icano(\obsm)$ (solid lines) and finite size LDF for $N=10$ (dot-dashed lines) and  $N=25$ (dashed lines). Parameters are for (a-b): $V=0.5$ (magenta lines) and $V=1.5$ (light blue lines). 
(c) and (d) Canonical LDF $\Icano(\obsm,a)$ of activity and magnetization and level lines (c) for $V=1.5$  
and (d) for $V=0.5$. Beige squares indicate the location of the stable solutions of Eqs.~(\ref{eq:MF-mag}-\ref{eq:MF-activity}) whereas dark-blue squares are for unstable solutions. \label{fig:ConvexLDF} For all figures $E=0$. For figure (c-d), the color map replaces the level lines for low values of the LDFs.
}
\end{figure*}

\subsubsection{Microcanonical LDF of magnetization and activity}

\label{sec:non-convexity-ldf}
As explain in sec.~\ref{sec:generalSetup}, the microcanonical LDF follows from the Legendre-Fenchel transform of the elements of the propagator of the generating function. Therefore Eq.~(\ref{eq:OrderOperationNonConvex2}) yields
\begin{equation}
  \label{eq:LDFMicrofromPGF}
  \Imicro(\obsmag,a) =\min\limits_{x_{\mathrm{f}},x_{\mathrm{i}}} \max\limits_{\kappa,\gamma}\left[ \kappa \obsmag + \gamma a - \lim\limits_{N,T \to \infty} \f{1}{NT} \ln  G(x_{\mathrm{f}},x_{\mathrm{i}},\kappa,\gamma) p_{\mathrm{i}}(x_{\mathrm{i}})\right]
\end{equation}
Each stationary solution of Eq.~\eqref{eq:stat_condition} denoted $\xetoile $ defines a partial LDF
\begin{equation}\label{eq:brancheI}
  I_{\xetoile}(\obsmag,a)=\max\limits_{\kappa,\gamma} \left\{\kappa \obsmag +\gamma a -  \mathcal{L}(\xetoile(\kappa,\gamma),0,\kappa,\gamma) \right\}.
\end{equation}
If the initial magnetization is in the same subpart of the state space than $\xetoile$, the system's fluctuations are best described by $I_{\xetoile}(\obsmag,a)$. When ensemble averaging on the initial condition, we look for the minimum on the stationary trajectories to obtain the microcanonical LDF
\begin{equation}
\Imicro(\obsmag,a) = \min\limits_{x \in \{ \xetoile\} } I_{x}(\obsmag,a).
\end{equation}
From the fact that $\Icano$ is the convex hull of $\Imicro$, we get the following inequality between the two LDFs:
\begin{equation}\label{eq:ComparaisonLDFMicroCano}
\Imicro(\obsmag,a) \geq \Icano(\obsmag,a).
\end{equation}
The microcanonical LDF $\Imicro(\obsm,a)$ is shown after the bifurcation on Fig.~\ref{fig:NonConvexLDF}c, We also provide in Fig.~\ref{fig:NonConvexLDF}a-b the partial LDF for activity and magnetization (after contraction) and their convex hulls. As expected, the microcanonical LDF is not convex: the ensemble equivalence does not hold (in a specific interval of magnetization and activity) for our model in the thermodynamics limit, due to lack of ergodicity.

Comparing the canonical LDF obtained from Eqs.~(\ref{eq:contractIconvex}) and the microcanonical LDF, we notice that the former is as expected the convex hull of the latter. 
\begin{figure*}[h]
\includegraphics[width=\columnwidth]{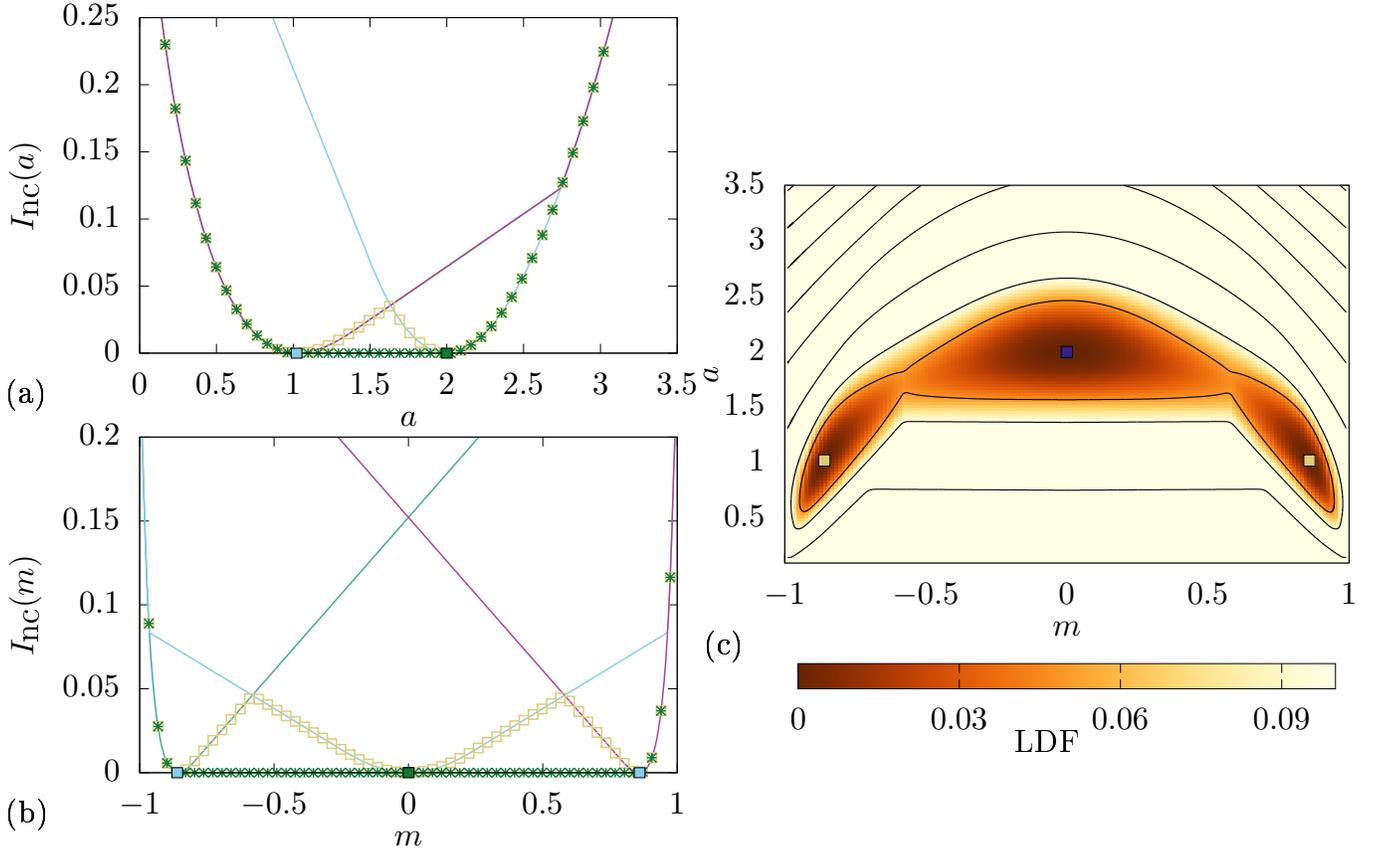}
\caption[Microcanonical LDF]{Branches of the microcanonical LDF of (a) activity and (b) magnetization (solid lines). Lines with symbols represent the canonical LDF (crosses) and the microcanonical LDF (squares). (c) Level lines of the microcanonical LDF of activity and magnetization. For all figures, parameters are: $E=0$ and $V=1.5$. For figure (c), the color map replaces the level lines for low values of the LDF. \label{fig:NonConvexLDF} }
\end{figure*}

\section{Non convex LDF and divergent mixing time}
\label{sec:LimitOrder}

In the previous sections, we have obtained the LDF of activity and magnetization for the canonical and microcanonical ensembles. In the thermodynamic limit, the ensemble equivalence is broken at high interaction even though it holds at finite size. We now explore the transition from equivalence to non-equivalence when increasing the size of the system, putting the emphasis on the order of the large size and long time limits in the computation of the statistics of magnetization and activity.

In this section, we point out the existence of a mixing time $t\e{mix}$ that depends on the system size. For systems of finite size, the mixing time governs the fluctuation regime. 
First, we define the mixing time from the spectral gap of the finite size operator $\bm{K}$. Second, we give an estimate of the mixing time for the Ising model and prove that it diverges when $N \rightarrow \infty$. This means that it is impossible to fully relax from the initial condition when the system includes too many spins, leading to an ergodicity breaking. 
Finally, we explore the different regime of fluctuations at finite size and time with numerical simulations, enlightening the coherence of our previous results.

For the numerical simulations of this section, the magnetic field is non zero ($E=0.2$) in order to break the up-down symmetry of the Ising model. Thanks to the magnetic field there exists two different stable magnetizations associated to different stable activities. Then, the activity probability density function will be bimodal as it is beyond numerical reach (at large $N$) to detect the third unstable magnetization and its associated activity. Without the magnetic field, just the unique stable activity would appear on the numerical simulations and the activity probability density function would be unimodal. For magnetization, we would observe only bimodal probability density function for symmetry reasons.

\subsection{Gap at finite size}
\label{sec:gap-at-finite}
We come back to the continuous-time process of Eq.~(\ref{eq:EqMaitress}) of Markov operator $\bm{K}$. We assume that $\bm{K}$ is diagonalizable. The eigenvalues are $\{ \lambda_i\}_{i=1,N}$, where $\lambda_1=0$ and all other eigenvalues are negative. Its set of eigenvectors is $\{l_i, \, r_i\}_{i=1,N}$ where $l_1$ is an uniform vector and $r_1=\pi$ is the stationnary probability. The spectral gap $\Delta \lambda$ of the operator $k$ is the difference between the largest eigenvalue $\lambda_1$ and the real part of its second eigenvalues $\lambda_2$.

Considering the initial probablity $p_0$, the probability
\begin{equation}
  \label{eq:formalSolutionMaitresseEq}
  p(x,t) = \left( e^{tk}p_{\mathrm{i}}\right)(x) =  \pi(x) + \sum_{i} e^{t \lambda_i} r_i(x) \sum_x l_i(x) p_{\mathrm{i}}(x)
\end{equation}
is a formal solution to Eq.~(\ref{eq:EqMaitress}).
Since the spectral gap $\Delta \lambda$ of the operator $\bm{K}$ is positive, the probability $p(x,t)$ converges towards $\pi$. The mixing time $t\e{mix} (\varepsilon)$ is used to quantify the time that the system takes to relax to the stationary probability \cite{Levin2008}. By definition, the mixing time  $t\e{mix} (\varepsilon)$ is the minimal time for which starting from any initial probability the system is at most at a distance  $\varepsilon$ of the stationary probability. Formally, the mixing time is
\begin{equation}
  \label{eq:mixingtimeDef}
  t\e{mix} (\varepsilon) \equiv \inf\left\{ t \geqslant 0\, :\,  \max\limits_{p_{\mathrm{i}} } \| e^{t\bm{K}}p_{\mathrm{i}} - \pstat \|_{TV} \leqslant \varepsilon  \right\},
\end{equation}
where the distance is the total variation distance defined as $\| u\|_{TV} = \sup_{A } \, u(A)  $. A common choice for $\varepsilon$ is $e^{-1}$, such that $t\e{mix} \equiv t\e{mix}(e^{-1})$. The mixing time quantifies the time needed to reach stationary probability, whatever is the initial probability. In particular, an infinite mixing time is a feature of non-ergodic systems. When considering large deviations statistics, we must formally take long time limit. We emphasize that for times longer than the mixing time, the large deviation statistics is approximately valid.

The mixing time of Eq.~(\ref{eq:mixingtimeDef}) is a maximization over all initial probability, therefore looking at $  \| e^{t\bm{K}}p_{\mathrm{i}} - \pstat \|_{TV} $ for an uniform initial probability $p_{\mathrm{i}}=1$ underestimates the mixing time. We plot on the Fig.~\ref{fig:MixingTime}c, the evolution of this distance for various sizes.
First, we observe an exponential scaling of the total variation distance with time. When comparing with Eq.~(\ref{eq:approx_proba_mixt}), we expect the mixing time to be connected with the spectral gap. Indeed we have for any Markov process~\cite{Levin2008}
\begin{equation}
  \label{eq:lienGapmixingtime}
  \f{1}{\Delta \lambda} \leqslant t\e{mix} \leqslant -\f{\log{p_{\mathrm{i}}^{\mathrm{min}}}}{\Delta \lambda}
\end{equation}
where $p_{\mathrm{i}}^{\mathrm{min}} = \min_{x} p_{\mathrm{i}}(x)$ and $\Delta \lambda$ is the spectral gap.
Second, at very short time, the evolution of total variation distance has a different scaling. This short time behavior is due to the other eigenvalues whose influence on the probability $p(x,t)$ disappears quickly. 
Finally, the evolution of the total variation distance with system size reveals the strong dependencies of the mixing time on the system size allowing for longer and longer transient behavior. 

At finite time, the fluctuations critically differ if $ T \gg t\e{mix}$ or if $ T \ll t\e{mix}$. On Fig.~\ref{fig:MixingTime}a and \ref{fig:MixingTime}b we show the empirical density probability for activity and magnetization for a system of $N=60$ spins and for different durations $T$, with $t\e{mix} \simeq 100$. And on Fig.~\ref{fig:MixingTime}d, we plot the empirical probability density of activity for a system size $N=200$ where $t\e{mix}\simeq 10^5$.

For $T \gg t\e{mix}$, the long time probabilities of magnetization and activity are those of an ergodic systems (the system has enough time to switch between MF states). The asymptotic probability is then given by the convex LDF of sec.~\ref{sec:conv-ldf-magn}. The convergence toward the plateau for these LDFs is discussed in the next section. Notice that 
the added small magnetic field is responsible of the unimodality of the probability density functions, see Fig.~\ref{fig:MixingTime}a and \ref{fig:MixingTime}b at $T=300$.

Otherwise when the second eigenvalue is well separated from the others, \textit{i.e.} when $\Delta \lambda \ll | \lambda_0-\lambda_3 |$, the probability $p(x,t)$ is well approximated for $| \lambda_0-\lambda_3 |^{-1} \ll T \ll t\e{mix}$ by
\begin{equation}
p(x,t) \simeq \pstat(x) +  e^{-t \Delta \lambda}r_1(x) \sum_y l_1(y) p_{\mathrm{i}}(y). \label{eq:approx_proba_mixt}
\end{equation}
It contains a term coming from the second eigenvalue. This second term induces the secondary peak on the probability density function. Therefore, for intermediate times before the mixing time, the system behaves as an effective non-ergodic system with fluctuations around each MF states (the system lacks of time to switch between states). We have a transient behavior where the probabilities of magnetization and activity are then bimodal, see Fig.~\ref{fig:MixingTime}d, and are linked with the microcanonical LDF of sec.~\ref{sec:non-convexity-ldf}.

\begin{figure*}[t]
\includegraphics[width=\columnwidth]{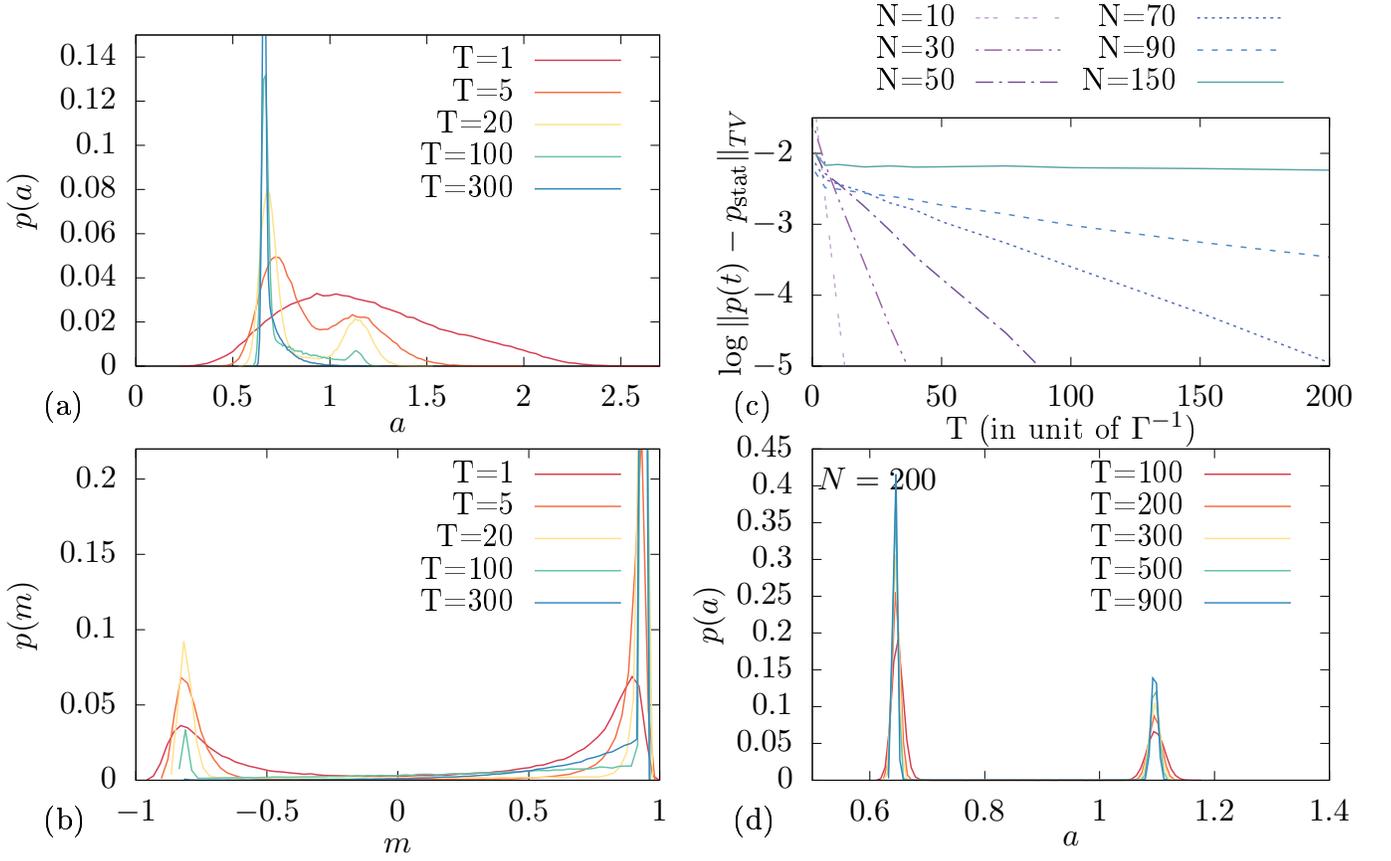}
\caption[Mixing time]{(a) Activity and (b) magnetization probability density functions obtained from numerical simulations of various duration $T$ for $N=60$. (c) Logarithm of the total variation distance between $   e^{t\bm{K}}p_{\mathrm{i}} $ and $\pstat$ in function of the duration of the evolution for various system size $N$. We simulate the evolution of $N$ spins using the Gillespie algorithm. For each size and time, a total of $15.10 ^{4}$ trajectories are drawn. The initial condition of the trajectory is draw from uniform distribution.
The probability density functions are computed from a histogram of $75$ bins between the minimum and  maximum values.
The total variation distance is the maximum of the difference between an histogram of the final value of the trajectory and an histogram of $15.10 ^{4}$ points drawn from probability of Eq.~(\ref{eq:probastat}).
(d) Activity probability density functions obtained from numerical simulations of various duration $T$ for $N=200$. The parameters are: $V=1.7$, $E=0.2$.   \label{fig:MixingTime}
}
\end{figure*}

\subsection{Estimation of the gap for systems with detailed balance}
\label{sec:estimation-gap-tree}

\begin{figure}[ht]
  \centering
  \includegraphics[width=\columnwidth]{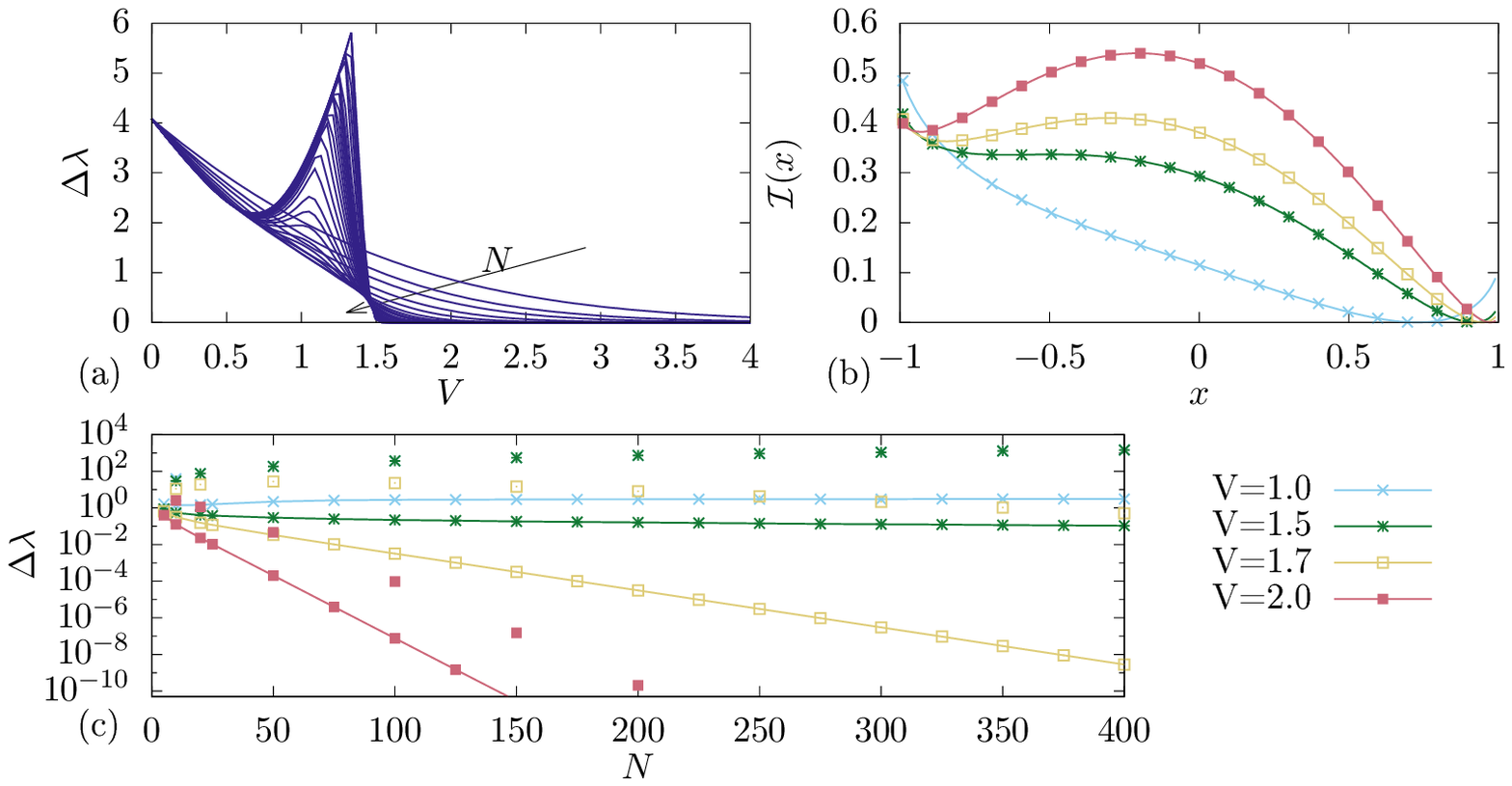}
  \caption[Gap]{(a) Spectral gap as a function of the interaction energy $V$ for various system size. (b) Large size LDF of $x$ defined by $\mathcal{I}(x) = \lim_{N \to \infty} \frac{1}{N} \ln \pi(x)$  for various interaction energy $V$. (c) Spectral Gap (lines with symbols) and bound of Eq.~(\ref{eq:boundonCheegerconstant}) (symbols) as a function of the system size for various interaction energy $V$. For all figures, we take $E=0.2$.}
  \label{fig:Gap}
\end{figure}

We now determine the scaling of the mixing time with the system size. To do so, we find an upper bound on the spectral gap of systems with detailed balance which, combined with the inequality of Eq.~(\ref{eq:lienGapmixingtime}), estimates the mixing time scaling.

For Markov processes with detailed balance, we can estimate the spectral gap from the Cheeger bound \cite{Lawler1988}. Indeed, in this case the Markov operator is symmetrizable, and we can use known results on the gap of symmetric matrices. We introduce the matrix $D_{\sqrt{\pi}}$ as the diagonal matrix with elements $D_{\sqrt{\pi}}(n,n) \equiv \sqrt{\pi_n}$, then the matrix $D_{\sqrt{\pi}} M D_{\sqrt{\pi}}^{-1}$ is symmetric if the Markov matrix $\bm{K}$ respects detailed balance. We introduce the Cheeger constant as
\begin{equation}
  \label{eq:cheegercst}
  \Phi = \inf_{\Xi \subset \Omega, \, 0 < \pi(\Xi)\leqslant 1/2} \f{Q(\Xi^c,\Xi)}{\pi(\Xi)}
\end{equation}
where  $\Xi$ is any subset of the set of state of our system $\Omega$ such that $\pstat(\Xi) = \sum_{x \in \Xi} \pstat(x) <1/2$ and,
\begin{equation}
  \label{eq:proba_low_subset}
  Q(\Xi^c,\Xi) = \sum_{x,y | x\in \Xi,\, y\in\Xi^c} \pstat(x) K({y,x})
\end{equation}
is the sum of probability flow from the subset $\Xi$ to the complementary of $\Xi$, denoted $\Xi^c$. From Ref.~\cite{Lawler1988}
\begin{equation}
  \label{eq:cheegerbounds}
  \Delta \lambda \leqslant 2 \Phi
\end{equation}
Computing the Cheeger constant is difficult in general, but it is possible to bound it. For a subset $\Xi$, the probability flow from $\Xi$ to $\Xi^c$ is bounded by
\begin{equation}
  \label{eq:bound_proba_flow_subset}
  Q(\Xi^c,\Xi) \leq N_b \pstat({\acute{x}}) K({\acute{y},\acute{x}})
\end{equation}
where $N_b$ is the number of edge connecting $\Xi$ and $\Xi^c$ and
\begin{equation}
  \label{eq:Cheeger_bound_edge_approx}
  (\acute{x}, \acute{y}) = \arg\max_{x \in \Xi,\, y \in \Xi^c }\pstat(x) K({y,x}),
\end{equation}
the edge supporting the biggest probability flow from $\Xi$ to $\Xi^c$. Denoting then $\tilde x$ the most probable state in $\Xi$, we have
\begin{equation}
  \label{eq:bound_Cheeger_1}
  \Phi \leqslant \f{N_b\pstat({\acute{x}}) K({\acute{y},\acute{x}})}{\pstat({\tilde x})}.
\end{equation}

Since we are mainly interested in the large size limit of the mixing time, we now assume that the stationary probability satisfies a large deviation principle
\begin{equation}
  \label{eq:LDF_size_cheeger}
  \pstat(x) \simeq_{N \to +\infty} e^{-N \mathcal{I}(x)}.
\end{equation}
We consider a connected subset of state $\Xi$ such that the state $x^\MF$ with $\mathcal{I}(x^\MF) = 0$ is not in $\Xi$. For large enough $N$, the probability of $\Xi$ is surely less than $1/2$.
Using the bound (\ref{eq:bound_Cheeger_1}) on the Cheeger constant, we find a large deviation estimate as
\begin{equation}
  \label{eq:boundonCheegerconstant}
  \Phi \leqslant \f{N_b\pstat({\acute{x}}) K({\acute{y},\acute{x}})}{\pstat({\tilde x})} \simeq N_b  K({\acute{y},\acute{x}}) e^{N(\mathcal{I}(\tilde x)-\mathcal{I}(\acute x))}.
\end{equation}
Therefore, if it exists $\Xi$ such that the LDF $\mathcal{I}(x)$ has a local minimum $\tilde x \in \Xi$ that is not a global minimum, we will have $\mathcal{I}(\tilde x)-\mathcal{I}(\acute x) <0 $. Then, if the product $ N_b  K({\acute{y},\acute{x}})$ is not diverging exponentially, we have bounded the spectral gap by something going to $0$ as $N \to + \infty$. Hence, the mixing time diverges with the system size, and the divergence is even exponential.

In our case, we consider the subset of state $\Xi^- = \{x | x \leqslant 0 \}$, if $E\geqslant 0$ it has a stationary probability less than $1/2$, otherwise if $E<0$, we consider the subset $\Xi^+ = \{x | x \geqslant 0 \}$.  We have then $N_b=1$, $\acute{x}=0^-$ and $\acute{y}=0^+$, the lower and upper closest states to $0$. The stationary probability respects a large deviation principle, with a local minimum appearing after the phase transition, see Fig.\ref{fig:Gap}b. Therefore, the mixing time is diverging after the phase transition, but not before.

These results are confirmed by numerical computation of the spectral gap. On Fig.~\ref{fig:Gap}a, we plot the spectral gap as a function of the system size $N$ and of the interaction energy $V$. Before the phase transition, the spectral gap remains finite and so does the mixing time. After the phase transition, the spectral gap goes to zero when increasing $N$. The speed of convergence of the spectral gap is well catched by the bounds~(\ref{eq:boundonCheegerconstant}) as indicates the exponential decay with increasing $N$ shown on Fig.~\ref{fig:Gap}c. The mixing time is so at least diverging exponentially with the system size.

\subsection{Mixing time as a criteria to study fluctuations}
From the previous results, we are now able to explain the effective transition from non-ergodicity to ergodicity as the observation time becomes longer than the mixing time. At the level of the dynamical ensembles, this corresponds to a transition from non-equivalence to equivalence of ensembles. The fluctuations differ before and after the mixing time that in addition is diverging with the system size $N$. In a first place, if we take the long time limit before large size limit, we overcome the mixing time and the system stays ergodic. As the equivalence remains valid, the fluctuations are correctly given by the canonical LDF. In a second place, if the large size limit is taken first, the mixing time is infinite such that we stay in the regime of fluctuations before the mixing time, the system is not ergodic and the LDF is likely to be non convex. Large deviation theory aims at describing asymptotic fluctuations, but the limits should not be understood strictly: it is fundamental to distinguish the mixing time and its dependency on the system size to interpret them correctly and use this framework for real experimental systems or numerical simulations. 

\section{Discussion and conclusion}
\label{sec:discussion}

Like for equilibrium ensembles, ergodicity breaking is fundamental to understand the non-equivalence of dynamical ensembles. Considering the connection between state variables and dynamical observables, equilibrium systems without equivalence of equilibrium ensembles will also exhibit non equivalence of dynamical ensembles. We have provided general arguments supporting that ergodicity is required, but may not be sufficient to observe non-equivalence between the microcanonical and canonical dynamical ensembles. We have considered a system without ergodicity (by construction) and one with an emergent non ergodicity. In both cases we have obtained LDFs with an explicit dependency in the initial condition, the so-called partial LDF. From this partial LDF, the microcanonical LDF follows from minimizing on the initial condition. For systems with emergent ergodicity breaking we have discussed the physical meaning of both the  microcanonical (non-convex) and canonical (convex) LDFs. We did so by comparing the observation time and the mixing time after which the initial condition can be forgotten.

At a different level of large deviation, one may investigate the structure of the trajectories corresponding to the null minima of the LDF and the relative weight of these minima. When the system is not ergodic and the LDF is non-convex, we observe separate zeros of the LDF that can be associated to the most probable values of the dynamical observable. In this case, trajectories remain in a subpart of the state space given by the initial condition. For the fully connected Ising model, the most probable states are the MF magnetizations and their relative probabilities are determined from the initial probability to start within each subpart of the state space.

On the contrary, when the system is ergodic and the LDF convex, we observe a whole continuous set of magnetizations and activities for which the LDF is zero, \textit{e.g.} the triangle region in Fig.~\ref{fig:ConvexLDF}c. These values corresponds to trajectories where the system spends a fraction of time around one MF state and another fraction around another MF state. Then, the final value of the observables is the time average on these trajectories. The move between MF states is allowed thanks to ergodicity. These trajectories are instantons of the Lagrangian, \textit{i.e.} trajectories going from a stationary solution of Euler-Lagrange equations to another~\cite{Freidlin1984,Onsager1953,Grafke2015}. In this paper, we did not take instantons into account for the minimization of Eq.~(\ref{eq:CGFfrompropagatorGF}). However, taking them into account does not modify the canonical LDF, but gives instead another LDF at a different large deviation speed~\cite{Nickelsen2018}. 

In our framework, we have assumed that it is always possible to find at least one stationary stochastic process (called driven process) that reproduces a fluctuation of the dynamical observables as a typical event. Ensemble equivalence means that this driven process is unique and does not depend on the initial condition. In our framework, ensemble non-equivalence means that there are several driven processes according to the initial conditions. We did not consider the case where no stationary driven process exists. This may happen for diffusive processes for which the state space is non-compact and the system relaxes while never reaching a stationary state~\cite{Szavits-Nossan2015,Nemoto2012}.

Model of glasses are known to have large relaxation time scales that are beyond numerical and experimental reach. The characterization of their fluctuations during this relaxation is an active field of research. Since we characterized fluctuations before and after mixing time, we expect that our framework could be useful to study metastable states of glasses.

\section*{Acknowledgement} 

We acknowledge Massimiliano Esposito for his advice on this project that started when GV was a post-doctoral fellow in his research team. 

\bibliography{These,ArXiv,Maths}

\appendix

\section{Derivation of the propagator for generating function}
\label{app:derivationPropagator}
We look for an expression of the propagator of the generating function $G(x_{\mathrm{f}},x_{\mathrm{i}},\kappa,\gamma)= \langle e^{NT(\kappa \obsmag+\gamma a)} \rangle_{x_{\mathrm{f}},x_{\mathrm{i}}}$ as a path integral over all trajectories for the system state $n$. We use the discrete time $t_k=k \dd  t$ corresponding to the $k$'s time step of duration $\dd   t$, with the final time $T=t_L = L \dd  t $. We write the system state $n_{t_k}=n_k$ for short. This state at time $t_k$ changes by $n_{k}-n_{k-1}=2\epsilon_k = -2, 0$ or $2$ when a spin jumps from up to down, does'nt jump or jumps from down to up. The sum over all paths for a given initial condition is
\begin{equation}
\sum_{[\epsilon_k]} = \sum_{\epsilon_0}\sum_{\epsilon_1} \cdots \sum_{\epsilon_{N-1}}.
\end{equation}
Since jump probabilities depends on the magnetization, the path integration should include a sum on the system state $n_k$ for all $k$
\begin{equation}
\sum_{[n_k]}= \sum_{n_1}\sum_{n_2} \cdots \sum_{n_{L-1}},
\end{equation}
with the condition that one spin at most jumps at each time step. In the above sums, the index $n$ goes from $-N$ to $N$. All in all, the propagator of the generating function writes
\begin{equation}
G(n_{\mathrm{i}},n_{\mathrm{f}},\kappa,\gamma)  =\sum_{[n_k,\epsilon_k]} \left [ \prod_{k=1}^{L} p(n_{k}|n_{k-1}) \delta(n_{k}-n_{k-1}-2\epsilon_{k-1}) \right ]  e^{\kappa \left( \sum_{k=1}^L n_{k-1} \dd t \right) +\gamma \left( \sum_{k=1\, ,\epsilon_k \neq 0}^L 1 \right) },
\end{equation}
with $p(n_{k}|n_{k-1}) $ the probability of the transition $ n_{k-1} \rightarrow n_{k}$. Let's use the Laplace representation of the Dirac distribution
\begin{equation}
\delta(n_{k}-n_{k-1}-2\epsilon_{k-1}) = \int_{i\bromwich_k-\infty}^{i\bromwich_k+\infty} \frac{\dd  \mathfrak{p}_{k}}{2  \pi}e^{ i \mathfrak{p}_{k}(n_{k}-n_{k-1}-2\epsilon_{k-1})}, \quad \bromwich \in \mathbb{R},
\end{equation}
for every integer $k$ between $1$ and $L$ to get 
\begin{equation}
G(n_{\mathrm{i}},n_{\mathrm{f}},\kappa,\gamma) = \sum_{[n_k,\epsilon_k]} \int \left ( \prod_{k=1}^{L}  \frac{\dd \mathfrak{p}_k}{2 \pi} \right ) 
 \left [ \prod_{k=1}^{L} p(n_{k}|n_{k-1}) e^{i \mathfrak{p}_k(n_{k}-n_{k-1}-2\epsilon_{k-1}) 
+\kappa n_{k-1} \dd t + \gamma}  \right ].
\end{equation}
In the next step, we sum over $[\epsilon_k]$ to obtain
\begin{multline}
G(n_{\mathrm{i}},n_{\mathrm{f}},\kappa,\gamma) =   \sum_{[n_k]}  \int \left (\prod_{k=1}^{L}  \frac{\dd  \mathfrak{p}_k}{2 \pi} e^{i \mathfrak{p}_k(n_{k}-n_{k-1})} \right ) \\
 \prod_{k=1}^{L} \left [ p(n_{k-1}|n_{k-1}) + \sum_{\epsilon=\pm 1}  p(n_{k-1}+\epsilon|n_{k-1}) e^{-2i\mathfrak{p}_k \epsilon 
 +\gamma }  \right ]e^{\kappa n_{k-1}\dd t}.
\end{multline}
The propagator when staying in the same state during a time step is $p(n_{k}|n_{k}) =  1-\linebreak (\sum\limits_{\epsilon=\pm1}K_{n_k+2\epsilon,n_k}) \dd  t$, and when changing of state writes $p(n_{k}+\epsilon|n_{k}) =  K_{n_k+2\epsilon,n_k} \dd  t$. Exponentiating the product yields
\begin{multline}
G(n_{\mathrm{i}},n_{\mathrm{f}},\kappa,\gamma) =  \sum_{[n_k]}   \int \left ( \prod_{k=1}^{L} \frac{\dd  \mathfrak{p}_k}{2 \pi}  \right )  \exp{  \left(i \dd t \sum_{k=1}^{L}  \mathfrak{p}_k \f{(n_{k}-n_{k-1})}{\dd t} \right) } \\
 \exp \left [ \sum_{k=1}^{L} \dd  t \left( \sum_{\epsilon} K_{n_{k-1}+2\epsilon,n_{k-1}} \left( e^{-2i \mathfrak{p}_k\epsilon +
 \gamma  } -1 \right)+\kappa n_{k-1} \right) \right ].
\end{multline}
The propagator of the generating function becomes in the continuous limit
\begin{equation}
G(x_{\mathrm{f}},x_{\mathrm{i}},\kappa,\gamma)=  \int \mathcal{D} [x] \int \mathcal{D}[\mathfrak{p}] 
 \exp \left[ N \int_{0}^{T} \dd  t   i \mathfrak{p}_t \dot x_{t} - \int_{0}^{T} \dd t \left(\sum_{\epsilon} J_\epsilon(x_{t}) \left(1- e^{-2 i \mathfrak{p}_t\epsilon 
 +   \gamma  } \right) +\kappa x_{t} \right) \right ]
\end{equation}
where $\int \mathcal{D} [x]$ (resp. $\int \mathcal{D} [\mathfrak{p}]$) is a short notation for the integral on the paths of density (resp. momentum) with given initial and final values.
Finally, the path integral expression of the generating function is given by
\begin{equation}
G(x_{\mathrm{f}},x_{\mathrm{i}},\kappa,\gamma)=  \int \mathcal{D} [x] \int \mathcal{D}[\mathfrak{p}] \exp \left( N \A_{\kappa,\gamma} [x,\mathfrak{p},T] \right )\label{eq:PathIntegralGF}
\end{equation}
with the action $\A_{\kappa,\gamma}$ being a functional of the paths $[x,\mathfrak{p}]$
\begin{equation}
\A_{\kappa,\gamma} [x,\mathfrak{p},T] = \int_{0}^{T} \dd  t   i \mathfrak{p}_t \dot x_{t}  - \int_{0}^{T} \dd t \Hr (x_{t},\mathfrak{p}_t,\kappa,\gamma). \label{eq:action}
\end{equation}
We have introduced the Hamiltonian function 
\begin{equation}
\Hr (x,\mathfrak{p},\kappa,\gamma) = \sum_{\epsilon=\pm 1} \left [ 1-\exp{ \left( -2 i \mathfrak{p} \epsilon 
+\gamma  \right) }  \right ] J_\epsilon(x) -\kappa x. \label{eq:Hamiltonien}
\end{equation}
This function is the representation in the continuous limit of a tridiagonal Metzler matrix, whose spectrum is real~\cite{*[{}][{This property is stated in problem 5, page 174.}]{Horn2012}}. As a consequence, $i \mathfrak{p}$ is a real number in the following.

\section{WKB approximation}
\label{app:WKBapprox}
We now use a saddle point integration to obtain the leading order in $N$ of the propagator of the generating function $G(x_{\mathrm{f}},x_{\mathrm{i}},\kappa,\gamma)$ which is asymptotically exact in the large size limit. For each integration, we deform the contour so that it goes through the saddle point 
of $\A_{\kappa,\gamma} [x,\mathfrak{p},t]$. Therefore, the saddle point is the minimum of $\A_{\kappa,\gamma} [x,\mathfrak{p},t]$ solving
\begin{align}
\f{\delta}{\delta \mathfrak{p}_t}\A_{\kappa,\gamma} [x,\mathfrak{p},T] = 0, \\
\f{\delta}{\delta x_t}\A_{\kappa,\gamma} [x,\mathfrak{p},T] = 0,
\end{align}
where $\f{\delta}{\delta x_t}$ denotes a functional derivative.

In this case, the saddle point calculation is equivalent to the WKB approximation of quantum mechanics. The propagator of the generating function will be given by the paths maximizing the action, i.e. the classical paths starting at $x_{\mathrm{i}}$ and ending at $x_{\mathrm{f}}$. These paths solve the Hamilton's equations
\begin{align}
i \dot x_t  =&  \,\partial_\mathfrak{p} \Hr (x_t,\mathfrak{p}_t,\kappa,\gamma),\label{eq:xEvol}  \\
i \dot{\mathfrak{p}}_t =& - \partial_x \Hr (x_t,\mathfrak{p}_t,\kappa,\gamma). \label{eq:pEvol}
\end{align}
The propagator of the generating function can be written as
\begin{equation}
  \label{eq:propagatorGF_WKB}
G(x_{\mathrm{f}},x_{\mathrm{i}},\kappa,\gamma) \simeq_{ N \to +\infty} \exp \left[ N\int_{0}^{T} \dd  t \left(  i \mathfrak{p}_t \dot x_{t}   - \Hr (x_{t},\mathfrak{p}_t,\kappa,\gamma) \right) \right],
\end{equation}
where $(x,\mathfrak{p})$ correspond to the solutions of Eq.~(\ref{eq:xEvol}-\ref{eq:pEvol})  with initial and final conditions $x_{\mathrm{i}}$ and $x_{\mathrm{f}}$. Moreover the Hamiltonian~(\ref{eq:Hamiltonien}) is time independent, and is hence a conserved quantity along the trajectory. Maupertuis' action $\int_{0}^{T} \dd  t  i \mathfrak{p}_t \dot x_{t}$ is a function of the Hamiltonian and writes
\begin{equation}
  \label{eq:maupertuisaction}
  \int_{0}^{T}\dd  t i \mathfrak{p}_t \dot x_{t} = \int_{x_{\mathrm{i}}}^{x_{\mathrm{f}}}\dd  x i \mathfrak{p}(x).
\end{equation}
where $\mathfrak{p}(x)$ can be obtained by inverting Eq.~(\ref{eq:Hamiltonien}). We observe that Maupertuis' action is always negative in view of the clockwise motion along the orbits
that can be justified by inspection of Eq.~(\ref{eq:xEvol}). 

In this large size limit, we can also use a saddle point integration on $\mathfrak{p}$ to switch to the Lagrangian framework. 
To do so, we look for the solution of Eq.~(\ref{eq:xEvol}). We derive a quadratic equation either for $e^{i\mathfrak{p}_{s}}$ or for $e^{-i\mathfrak{p}_{s}}$ whose solutions are
\begin{equation}
e^{\pm 2i\mathfrak{p}_{s}}=\f{\mp \dot x_{s} + \sqrt{\dot{x_{s}}^2+\varphi(x_{s},\gamma)}}{2 J_\mp(x_{s}) e^{\gamma}},
\end{equation}
with $\varphi(x,\gamma)$ defined in Eq.~\eqref{eq:varphi}. Inserting this into Eq.~\eqref{eq:PathIntegralGF} leads to the propagator of the generating function of Eq.~\eqref{eq:propagatorGFclass} with Lagangian~\eqref{eq:Lagrangien}.

\end{document}